\pdfoutput=1 
\documentclass[11pt]{article}
\usepackage{amsmath,amssymb,amsthm,mathtools}
\usepackage[margin=1in]{geometry}
\usepackage{graphicx}
\usepackage{xurl}
\usepackage[hidelinks]{hyperref}
\hypersetup{pdftitle={Delayed Verification in Multi-Agent LLM Systems: Stability of a Consensus Model and Coherence-Based Corrector Placement},pdfauthor={Igor Itkin}}
\usepackage{microtype}
\usepackage[ruled,vlined,linesnumbered]{algorithm2e}
\newtheorem{theorem}{Theorem}
\newtheorem{proposition}{Proposition}
\newtheorem{lemma}{Lemma}
\newtheorem{corollary}{Corollary}
\theoremstyle{definition}
\newtheorem{assumption}{Assumption}
\newtheorem{remark}{Remark}

\newcommand{\R}{\mathbb{R}}
\newcommand{\Lg}{L_{\mathrm g}}
\newcommand{\spec}{\operatorname{spec}}
\DeclareMathOperator{\diag}{diag}
\DeclareMathOperator{\tr}{tr}

\title{\bf Delayed Verification in Multi-Agent LLM Systems:\\
Stability of a Consensus Model and Coherence-Based Corrector Placement}
\author{Igor Itkin\\
\small Independent Researcher, Tel Aviv, Israel\\
\small \texttt{ig.itkin@gmail.com}\\
\small ORCID: \href{https://orcid.org/0009-0004-9513-8463}{0009-0004-9513-8463}}
\date{June 2026; revised September 16, 2026}

\begin{document}
\maketitle

\begin{abstract}
Multi-agent large language model (LLM) systems often rely on verifier and critic agents to suppress hallucinations, but
verification is delayed. During this delay, false claims can propagate through the agent network. We
model this process as delayed consensus on a graph with grounded corrector nodes. Spectral decomposition
by the grounded Laplacian yields an exact stability threshold for the verification dose: correction
that is too strong or too delayed can destabilize this linear model through an oscillatory mode.
This does not establish a bounded nonlinear cycle. When communication and
verification have the same delay, their gains add; at delay two the threshold for the dimensionless
combined dose $\eta(\mu+\kappa)$ is the inverse golden ratio. A related coherence surrogate gives a supermodular placement objective: greedy
assignment achieves at least $(1-1/e)$ of the optimal coherence reduction for a fixed corrector budget.
This guarantee concerns the surrogate, not the mean-square error under constant random forcing. Experiments
with an externally computed correction law show larger signed-error variability at longer effective
lag in five of six tested models in an exploratory all-run analysis. A separate expanded factual
study uses $400$ new questions, corrected delay indexing and complete response logs. After semantic
review, conservative completion bounds still permit both signs of the change in binary error
amplitude, largely because many states contain abstentions. A post-hoc analysis finds lower
abstention frequency at the longer lag on TruthfulQA, without establishing a change in error
variability. Neither study establishes general factual stability; non-negativity alone does not
rule out oscillations under persistent forcing.
\end{abstract}


\section{Introduction}

Large language models hallucinate, and in multi-agent systems hallucination stops being a static
property of one output and becomes a dynamic process: claims are exchanged, revised, and
reused as context, so an unsupported claim from one agent can be amplified by others
\cite{HallCascade,SparkFire,Collective}. The phenomenon is sharp even inside a single model:
Zhang et al.\ show that once an LLM commits to a wrong answer it generates further false claims to
justify it (``hallucination snowballing'') even when it can separately recognize them as wrong
\cite{Snowball}. Approaches to improving responses include self-feedback
\cite{SelfRefine}, reinforced self-checking \cite{MARCH}, and multi-agent debate
\cite{Du,ReConcile}. These mechanisms differ in their access to evidence: for example,
Self-Refine uses the same LLM to generate, critique and refine an answer, without an oracle
truth signal. Evidence-grounded verification is a distinct design choice.

Verification, however, has latency. A verifier reads a claim, retrieves evidence, and returns a
correction only after several interaction steps; meanwhile the unverified claim has already
propagated. Delayed negative feedback is the classic ingredient of oscillation and instability in
control systems. We study the following question in a specified consensus model: \emph{can the
very act of verification, if delayed, destabilize the factuality it is meant to protect, and how
should correctors be dosed and placed to avoid this?}

The agent-safety literature is moving from post-hoc analysis toward online monitoring of
cascades (causal cross-channel monitoring \cite{CASPIAN}, online failure auditing
\cite{AgentForesight}, temporal-graph anomaly detection \cite{GUARDIAN}). These monitoring objectives
are distinct from the stability of a controller that acts on delayed feedback. Quickest-change-detection
theory studies detection delay under false-alarm constraints \cite{Page,Lorden,Pollak,VeeraBanerjee};
here we analyze correction gain and delay in a postulated network recurrence. The single-agent
detection problem is treated separately in our companion work \cite{ItkinQCD}.

Delayed institutional regulation of adaptive multi-agent systems provides related motivation
\cite{Itkin2026}; that work reports a supercritical Hopf bifurcation in a different nonlinear
model. That nonlinear conclusion does not transfer to our discrete linear recurrence.
Existing dynamical accounts of LLM interaction include
consensus models \cite{OpinionConsensus,ChenConsensus}, hidden-anchor deliberation
\cite{HiddenAnchors}, self-correction as feedback control \cite{SelfCorrFeedback}, and models of
factual-error propagation \cite{HallCascade,Collective,SparkFire}. Our specific object is a discrete
network recurrence with an explicitly delayed restoring term. Its stability threshold is a property
of that recurrence; relating it to ordinary factual verification requires independent identification.
Empirical debate failures \cite{DebateCollapse} motivate this question but do not establish that
this mechanism explains those failures.

\noindent This paper makes five contributions.
\begin{itemize}
  \item \textbf{A model and a reduction.} We cast the corrected multi-agent loop as a delayed consensus
        with grounded correctors that the grounded Laplacian decouples into independent scalar delay
        recurrences (Lemma~\ref{lem:decouple}; Sections~\ref{sec:model}--\ref{sec:stab}).
  \item \textbf{An exact verification-dose limit.} On the branch $0<a<1$, increasing the
        dimensionless dose $\beta=\eta\kappa$ first destabilizes the linear loop through a complex
        pair. The threshold decreases with delay and approaches the inverse golden ratio at
        delay two as $a\uparrow1$; it is not a theorem of bounded nonlinear oscillation (Theorem~\ref{thm:main},
        Proposition~\ref{prop:boundary}, Lemma~\ref{lem:cheb}, Corollary~\ref{cor:dose-explicit};
        Section~\ref{sec:stab}).
  \item \textbf{Placement for a coherence surrogate.} The resolvent trace is supermodular, so greedy
        placement achieves at least $1-1/e$ of the optimal reduction in this metric. The guarantee
        does not transfer directly to mean-square error under constant forcing (Theorem~\ref{thm:place};
        Section~\ref{sec:place}).
  \item \textbf{Two coupled delays.} We derive the unit-circle locus of the two-delay characteristic
        polynomial, including its singular cases. Equal delays give an exact combined-gain ceiling,
        equal to the inverse golden ratio at delay two (Theorem~\ref{thm:twodelay}, Corollary~\ref{cor:equal};
        Section~\ref{sec:two}).
  \item \textbf{An empirical regime comparison.} Numeric-estimation debates driven by an external
        controller exhibit greater signed-error variability at longer lag in an exploratory
        all-run comparison across several models. The archived experiment requires an indexing
        correction and separate conditional/all-run analyses. A new $400$-question factual study
        records all responses and uses corrected indexing, but abstentions and residual scoring
        ambiguity leave the direction of its binary-amplitude contrast unidentified. A separate
        analysis distinguishes the frequency of each response outcome from its temporal variability;
        non-negativity alone does not establish stability
        (Remark~\ref{rem:absorb}; Section~\ref{sec:emp}).
\end{itemize}

Project repository: \url{https://github.com/YehudaItkin/delayed-verification-llm}.
The project contains numerical checks, archived trajectories and reanalysis scripts. Historical
archives lack raw responses for some experiments, preventing independent reconstruction of their
factual scores. The expanded study instead preserves full request/response logs, a frozen manifest,
versioned semantic decisions and checksummed analyses (Section~\ref{sec:expanded}).
The data layout and offline reproduction commands are given in Appendix~\ref{app:repro}.

The theoretical argument has three steps: define the recurrence (Section~\ref{sec:model}),
derive its scalar stability threshold (Section~\ref{sec:stab}), and distinguish the placement
objective from stability (Section~\ref{sec:place}). Section~\ref{sec:two} adds communication
delay. Section~\ref{sec:emp} then asks what synthetic and LLM experiments can establish about
these models; its observed variability measures are not tests of asymptotic stability.

\section{Related work}\label{sec:related}

We relate empirical reliability studies of multi-agent LLMs to classical delayed-feedback control.
The stability and placement ingredients have precedents in control theory; the contribution is their
specialization to a stated verification-loop model, with explicit limits on its empirical interpretation.

On the LLM side, Jamshidi et al.\ track claim-level inconsistency across agent chains and report
that deeper chains can reduce an explicit hallucination score while eroding factual accuracy
\cite{HallCascade}; their related work models collective hallucination and interventions
\cite{Collective}. Xie et al.\ model error cascades and study provenance-aware mitigation
\cite{SparkFire}. Evaluator preferences and unsafe content are studied as distinct propagation
phenomena \cite{Contagion,PropGuard}. These are already dynamical accounts, rather than a purely
static literature. Debate can improve factuality \cite{Du,ReConcile}; Liang et al.\ propose it as a
remedy for degeneration during \emph{self-reflection}, while also discussing judge bias \cite{Liang}.
Other work separates the contribution of voting from interaction \cite{DebateVote}, studies
sycophantic consensus \cite{TalkCheap}, and varies communication topology \cite{SparseTopo}.

Related dynamical models include DeGroot-inspired consensus \cite{OpinionConsensus},
hidden-anchor deliberation \cite{HiddenAnchors}, and single-agent error-transition models
\cite{SelfCorrFeedback}. Delay-free averaging need not make every agent's opinion monotone:
a cooperative two-node averaging matrix with diagonal $0.1$ and off-diagonal $0.9$ has disagreement
multiplier $-0.8$, giving damped alternation. Delay robustness also depends on the update rule.
Yao and Li \cite{YaoLi} retain a current self-state and delay normalized neighbor terms;
our two-delay model delays the entire Laplacian term, including its diagonal. Their convergence
result therefore cannot establish delay robustness of our gossip term. Indeed, with no verification,
$p=20/27$ and communication delay $d=2$, the scalar polynomial is
\[
 z^3-z^2+20/27=(z+2/3)(z^2-5z/3+10/9).
\]
Its conjugate roots have squared modulus $10/9>1$, whereas instantaneous gossip has multiplier
$7/27$. Thus delayed gossip alone can destabilize our model. Sections~\ref{sec:stab} and
\ref{sec:two} distinguish instantaneous gossip from coupled communication and verification delays.

Detection and control address different questions. Failure attribution benchmarks
\cite{AgentHallu,WhoWhen,TRAIL} and online monitoring methods \cite{AgentForesight,CASPIAN,GUARDIAN}
locate or anticipate failures; we analyze a specified correcting action. Sequential detection
provides related tools, including CUSUM and minimax delay criteria
\cite{Page,Lorden,Pollak,Moustakides,Lai,VeeraBanerjee,XieQCD,TNB}, but no detection-delay
result is proved here. Potential LLM monitoring signals include semantic uncertainty and
self-consistency \cite{SemEntropy,SelfCheck,HallSurvey}; confidence-based early exit
\cite{EntropyCoT} and social-learning interventions \cite{JainKrishnamurthy} have different objectives.
Our stability boundary specializes the discrete delay equation studied by Kuruklis \cite{Kuruklis}.

The placement
results adapt established network-control objectives. Clark et al.\ study supermodular
\emph{transient convergence error} \cite{ClarkLeader}; steady-state variance under noise-corrupted
leaders is treated directly by Mackin and Patterson \cite{MackinPatterson}. These metrics should
not be conflated with constant-forcing squared error. Our resolvent-trace proof and greedy reduction
guarantee specialize this line of work \cite{ClarkSubmod,PiraniSundaram,Nemhauser}.
Sherlock \cite{Sherlock} instead studies budgeted verification in agentic workflow execution;
that setting differs from stability and coherence in our symmetric consensus recurrence.
The delay axis has a control precedent as
well: Pirani et al.\ characterize the maximum-allowable-delay margin of leader--follower consensus
through the grounded-Laplacian spectrum and turn it into a leader-placement centrality
\cite{PiraniRobust}, but in continuous time and with a single delay; we add the discrete verification
dose, the two coupled delays, and the LLM loop. The grounding signals themselves (claim
verification and faithfulness \cite{FEVER,RAG}) are standard.

\section{Model}\label{sec:model}

We abstract a round of multi-agent deliberation as follows. Each agent carries a scalar \emph{belief}
about a claim, represented by a real coordinate rather than a calibrated probability, and advances in
rounds: agents revise their belief toward those of their neighbours (debate), while designated
\emph{corrector} agents, which have access to verified evidence, push beliefs back toward the ground
truth. Two features of deployed systems drive the dynamics and are the focus of this paper. First,
correction is \emph{delayed}: a verifier reads a claim, retrieves evidence, and replies only after a
latency, so it acts on a stale belief. Second, persistent faults can pull beliefs away from truth;
we represent that influence by a constant additive forcing. The problem we study is then stated as:
given the interaction graph, the verification delay, and the faulty forcing, (a) when does the loop
converge to its generally biased equilibrium, and when does it lose stability, and (b) where should a limited
budget of correctors be placed to reduce a specified coherence surrogate? The rest of this section makes the model
precise; Sections~\ref{sec:stab}--\ref{sec:two} answer~(a) and Section~\ref{sec:place} answers~(b).

Let $\mathcal G=(V,E)$ be a finite undirected weighted graph on $V=\{1,\dots,n\}$.
Its adjacency matrix has $A_{ij}=A_{ji}\ge0$, $A_{ii}=0$, and an edge exactly when $A_{ij}>0$;
$D=\diag(A\mathbf1)$ and $L=D-A\succeq0$ are its degree matrix and Laplacian.
Here $\mathbf1$ denotes the all-ones vector, $I$ the identity of the relevant dimension,
and $B\succeq0$ ($B\succ0$) means that the symmetric matrix $B$ is positive semidefinite
(positive definite). We use $\tr$ for matrix trace and $\|\cdot\|_2$ for the Euclidean
vector norm or its induced matrix norm. Agent $i$
holds a scalar belief $b_{i,t}\in\R$ about a claim with ground truth $b^\star$; the error is
$e_{i,t}=b_{i,t}-b^\star$. A \emph{corrector} set $R\subseteq V$ ($|R|=m$) consists of agents with
access to verified evidence, held at truth ($e_{i,t}=0$, $i\in R$), an ideal hard-pinning abstraction related to pinning control
\cite{WangChen}. The remaining \emph{free} agents
$V_{\mathrm f}=V\setminus R$ ($n_{\mathrm f}=n-m\ge1$) carry error $e_t\in\R^{n_{\mathrm f}}$; a faulty
subset injects a constant bias collected as a forcing $g\in\R^{n_{\mathrm f}}$.
Symmetry means reciprocal weights, a modeling assumption rather than a property inferred from
the exchange of peer text in our experiments. It enables the real orthogonal decoupling of
Lemma~\ref{lem:decouple}. Directed interaction matrices need not be symmetric and may be non-normal;
the real-spectrum thresholds below do not cover them in general.
The hard pins remain at truth instantaneously. Separately, every free node receives the uniform
delayed restoring term specified below. Its gain is not derived from the number or incident edges
of the hard pins. This restoring signal is oracle-valued in the mathematical model; the experimental
LLM verifier need not be correct. Content-level phenomena such as truth dilution and claim
provenance are outside the scalar abstraction.

Free agents run consensus (step $\eta>0$) plus a uniform \emph{delayed verification} of gain
$\kappa>0$ and integer latency $\delta\ge1$:
\begin{equation}\label{eq:dyn}
  e_{t+1} \;=\; (I-\eta \Lg)\,e_t \;-\; \eta\kappa\, e_{t-\delta} \;+\; \eta g,
\end{equation}
where $\Lg \coloneqq L[V_{\mathrm f},V_{\mathrm f}]$ is the principal submatrix of $L$ on the free
nodes (the \emph{grounded Laplacian}). Reading \eqref{eq:dyn} term by term: $(I-\eta\Lg)\,e_t$ is the
consensus (debate) step, in which each free agent moves toward its neighbours at rate $\eta$;
$-\eta\kappa\,e_{t-\delta}$ is the separate uniform restoring force toward truth, with \emph{strength}
$\kappa$ and \emph{delay} $\delta$ (it acts on the $\delta$-step-old error); and $\eta g$ is the
constant pull of the faulty agents. The verification strength $\kappa$ and delay $\delta$ are the two
control knobs, and everything that follows analyses \eqref{eq:dyn} as a function of them.

The initial history $e_{-\delta},\ldots,e_0$ is arbitrary and finite, and updates start at $t=0$.
We call the homogeneous loop \emph{asymptotically stable} if every such history converges to zero;
for this finite-dimensional linear recurrence this is equivalent to every characteristic root
lying strictly inside $|z|<1$. With constant $g$, stability instead means convergence to the
equilibrium after subtracting its constant history. A nonzero equilibrium error is therefore
compatible with stability. At update $t\to t+1$, delay $\delta$ means feedback from state
$t-\delta$; the empirical immediate-feedback condition uses the extension $\delta=0$.

The additive fault abstraction is not identical to eliminating fixed erroneous peers from a
consensus graph. If external peers have fixed error vector $f$ and weights $B\ge0$ into the free
nodes, their contribution is $\eta Bf-\eta\diag(B\mathbf1)e_t$. Thus exact elimination replaces
$\Lg$ by $\Lg+\diag(B\mathbf1)$ as well as setting $g=Bf$. Omitting this diagonal term changes
the stability problem. Equation~\eqref{eq:dyn} instead specifies $\Lg$ and additive $g$ directly;
we do not infer either operator from the fixed textual peers in the factual experiments.

To make the grounded Laplacian positive definite, we impose the following anchoring condition.
\begin{assumption}[grounding reachability]\label{ass:reach}
Every connected component of the subgraph induced by $V_{\mathrm f}$ contains a node adjacent to $R$.
\end{assumption}
\noindent This condition is equivalent to $\Lg\succ0$. To see this, extend a free-node vector
$v$ by zeros on $R$ to obtain $\tilde v$. Then
\[
 v^\top\Lg v=\sum_{\{i,j\}\in E} A_{ij}(\tilde v_i-\tilde v_j)^2.
\]
Under the condition, zero energy forces $v=0$; without it, a nonzero constant on an unanchored
free component has zero energy. We assume the condition for the hard-grounded stability results;
Section~\ref{sec:place} states separately when an ungrounded baseline is allowed.
It is not necessary for stability of every corrected loop: the uniform verification term
in \eqref{eq:dyn} can also stabilize a zero Laplacian mode. For example, at $\delta=1$ and
$\eta\kappa=1/2$, that mode has characteristic polynomial $z^2-z+1/2$, with roots
$(1\pm i)/2$ of modulus $1/\sqrt2<1$. Moreover, for $\kappa>0$, the operator $\Lg+\kappa I$
is invertible even if $\Lg$ is singular. The unique equilibrium is
$e_\infty=(\Lg+\kappa I)^{-1}g$, which is truth when $g=0$; convergence to it requires the stability
conditions below.

\section{Stability and the verification dose}\label{sec:stab}

\begin{lemma}[grounded-Laplacian decoupling]\label{lem:decouple}
Let $\Lg = Q\diag(\mu_1,\dots,\mu_{n_{\mathrm f}})Q^\top$, where $Q$ is orthogonal and
$0<\mu_1\le\cdots\le\mu_{n_{\mathrm f}}$ are the eigenvalues.
Put $x_t = Q^\top e_t$, $\hat g=Q^\top g$. Since $\kappa I$ is scalar (hence commutes with $\Lg$),
\eqref{eq:dyn} decouples into independent scalar recurrences
\begin{equation}\label{eq:mode}
  x_{i,t+1} = a_i\,x_{i,t} - \eta\kappa\, x_{i,t-\delta} + \eta\hat g_i,
  \qquad a_i \coloneqq 1-\eta\mu_i .
\end{equation}
\end{lemma}
\begin{proof}
Multiply \eqref{eq:dyn} by $Q^\top$ and use $Q^\top Q=I$ and
$Q^\top\Lg Q=\diag(\mu_i)$. Both the instantaneous and delayed terms are then diagonal
in the same coordinates, giving \eqref{eq:mode}.
\end{proof}

The decoupling is the crux of the analysis. An $n_{\mathrm f}$-dimensional \emph{delayed} network,
ordinarily an awkward object, collapses into independent scalar modes, so the entire stability
question reduces, after subtracting $x_{i,\infty}=\hat g_i/(\mu_i+\kappa)$, to the homogeneous
scalar delay recurrence $x_{t+1}-a_i x_t+\beta\,x_{t-\delta}=0$ (with
$\beta=\eta\kappa$), one copy per eigenvalue $\mu_i$ of $\Lg$. We analyse this scalar recurrence once
and then quantify the result over the spectrum.

\begin{proposition}[explicit $\delta=1$ region]\label{prop:delta1}
For $\delta=1$, mode \eqref{eq:mode} is asymptotically stable iff
$\eta\kappa<1$ and $\mu_i<2/\eta+\kappa$.
\end{proposition}
\begin{proof}
Characteristic $z^2-a_i z+\eta\kappa$ with $a_i=1-\eta\mu_i$; Jury's conditions for
$z^2+c_1z+c_0$ ($c_1=-a_i$, $c_0=\eta\kappa$) reduce to $\eta(\mu_i+\kappa)>0$ (automatic),
$\mu_i<2/\eta+\kappa$, and $\eta\kappa<1$.
\end{proof}

\begin{theorem}[networked stability]\label{thm:main}
Under Assumption~\ref{ass:reach}, the homogeneous part of \eqref{eq:dyn} is asymptotically stable iff
$(a_i,\eta\kappa)\in\mathcal D_\delta$ for every $\mu_i\in\spec(\Lg)$, where $\spec$ denotes the
eigenvalue set and
\[
 \mathcal D_\delta=\{(a,\beta)\in\R\times(0,\infty):
       \text{all roots of }z^{\delta+1}-az^\delta+\beta\text{ have }|z|<1\}.
\]
For $\delta=1$ this is
$\eta\kappa<1$ and $\mu_{\max}(\Lg)<2/\eta+\kappa$. The positive-mode branch for larger delays is treated separately below.
\end{theorem}
\begin{proof}
By Lemma~\ref{lem:decouple} the system is the direct sum of the scalar modes; it is stable iff each
is. Apply Proposition~\ref{prop:delta1}.
\end{proof}

For $\delta=1$, stability requires both a verification gain below $1/\eta$ and a step size satisfying
the fastest-mode constraint. Even single-step verification loses asymptotic stability at excessive
gain. We next track the \emph{onset} of instability for general delay $\delta$ on the non-negative
mode branch used in Corollary~\ref{cor:dose-explicit}.

\begin{proposition}[oscillatory boundary for positive $a$]\label{prop:boundary}
Fix an integer $\delta\ge1$ and the mode $x_{t+1}-a x_t+\beta x_{t-\delta}=0$ with $\beta=\eta\kappa>0$,
$0<a<1$. As $\beta$ grows from $0$, the first roots to reach the unit circle are a simple
complex-conjugate pair $e^{\pm i\theta_\star}$. The critical dose is
\begin{equation}\label{eq:boundary}
  a=\frac{\sin((\delta+1)\theta)}{\sin(\delta\theta)},\qquad
  \beta_c=\frac{\sin\theta}{\sin(\delta\theta)},\qquad
  \frac{\pi}{2\delta+1}<\theta<\frac{\pi}{\delta+1},
\end{equation}
where the branch contains a unique $\theta=\theta_\star$ for each $a\in(0,1)$.
The mode is asymptotically stable iff $0<\beta<\beta_c$, so
$\kappa_{\max}(a,\delta)=\beta_c/\eta$. The critical pair has angular frequency
$\theta_\star$ radians per round (nominal period $2\pi/\theta_\star$ rounds).
The sampled mode need not have an integer period.
\end{proposition}
\begin{proof}
See \ref{app:proofs}.
\end{proof}

At $a=0$, the polynomial is $z^{\delta+1}+\beta$: all $\delta+1$ roots have modulus
$\beta^{1/(\delta+1)}$ and reach the unit circle simultaneously at $\beta=1$. Thus the threshold
extends to $a=0$, but the assertion of a first isolated pair does not. Negative $a$ requires a separate
boundary analysis. For example, at $a=-1/2$, $\delta=2$ and $\beta=1/2$,
\[
  z^3+\tfrac12 z^2+\tfrac12=(z+1)(2z^2-z+1)/2.
\]
The first crossing is at $z=-1$; the other roots $(1\pm i\sqrt7)/4$ have modulus $1/\sqrt2$.
For $0<\beta<1/2$, stability follows from $|z^2(z+1/2)|\ge1/2>\beta$ on $|z|=1$
and Rouch\'e's theorem. This case is outside Proposition~\ref{prop:boundary}.

Equation~\eqref{eq:boundary} specifies the boundary parametrically. Chebyshev polynomials give
an equivalent algebraic equation and a direct proof of the monotonicity needed for a design rule.

\begin{lemma}[Chebyshev form; the dose is monotone]\label{lem:cheb}
Let $U_0(c)=1$, $U_1(c)=2c$ and $U_{j+1}(c)=2cU_j(c)-U_{j-1}(c)$ be the Chebyshev
polynomials of the second kind. With $c=\cos\theta$ and
$\sin(\delta\theta)/\sin\theta=U_{\delta-1}(c)$, the boundary \eqref{eq:boundary} is algebraic,
\[
  a=\frac{U_\delta(c)}{U_{\delta-1}(c)},\qquad \beta_c=\frac{1}{U_{\delta-1}(c)},
  \qquad c\in\Big(\cos\tfrac{\pi}{\delta+1},\,\cos\tfrac{\pi}{2\delta+1}\Big)\ (a\in(0,1)),
\]
with $\beta_c(a,1)\equiv1$, $\beta_c(0,\delta)\equiv1$, and the exact value
$\beta_c(a,2)=\tfrac12(\sqrt{a^2+4}-a)$ for $a\in[0,1)$. For each $\delta\ge2$ the dose is
strictly decreasing in $a\in(0,1)$; at $\delta=1$ it is constant. For each fixed $a\in(0,1)$,
it is strictly decreasing in the integer delay $\delta\ge1$, with $\beta_c<1$ for $\delta\ge2$.
At $a=0$ it remains $1$ for every delay. The limit as $a\uparrow1$ is
$\beta_c(1,\delta)=2\sin(\pi/(4\delta+2))$.
\end{lemma}
\begin{proof}
See \ref{app:cheb}.
\end{proof}

On the branch $a=1-\eta\mu\in[0,1)$,
the ceiling does not increase as the delay $\delta$ grows and does not decrease as $\mu$ grows. Thus the slowest
mode attains the binding constraint (all modes share the same ceiling when $\delta=1$),
which the corollary now pins down. Two remarks sharpen the claim. First, because \eqref{eq:boundary}
is a reparametrization of the exact, necessary-and-sufficient stability region of Kuruklis
\cite{Kuruklis}, $\kappa_{\max}(a,\delta)$ is not merely a sufficient bound but the \emph{true}
boundary on the stated branch; it is the discrete-time, delay-difference analogue of the continuous consensus delay margin
$\tau^\star=\pi/(2\lambda_{\max})$ of Olfati-Saber and Murray \cite{OlfatiSaber}. Second, the integer
$\delta$ is the operative case: the proof in~\ref{app:cheb} establishes both monotonicity statements for
arbitrary integer delays, with the boundary cases stated separately.

\begin{corollary}[dose limit and binding mode]\label{cor:dose-explicit}
Assume $\eta\mu_{\max}(\Lg)\le1$ (so each $a_i\in[0,1)$). Then the network dose limit is set by the
\emph{slowest} grounded mode,
$\kappa<\kappa_{\max}\big(1-\eta\mu_{\min}(\Lg),\delta\big)$.
For nested hard-corrector sets $R\subseteq R'$ with at least one free node remaining, Cauchy
interlacing gives $\mu_{\min}(L_{\mathrm g}(R'))\ge\mu_{\min}(L_{\mathrm g}(R))$ and
$\mu_{\max}(L_{\mathrm g}(R'))\le\mu_{\max}(L_{\mathrm g}(R))$. At fixed $\eta$, adding hard
correctors therefore preserves the assumed step-size condition and \textbf{does not lower the
delay-induced dose ceiling}; the ceiling is unchanged at $\delta=1$. In
particular, reaching the dose ceiling loses asymptotic stability, and exceeding it gives a growing
linear mode. Increasing delay can also cross the ceiling when the binding $a$ is positive.
The oscillatory loss of stability is the regime we call
\emph{verification-induced oscillation}.
\end{corollary}
\begin{proof}
The largest $a_i=1-\eta\mu_i$ corresponds to $\mu_{\min}$, and Lemma~\ref{lem:cheb}
makes its threshold no larger than any other mode's. Proposition~\ref{prop:boundary}
and the $a=0$ case give the network criterion. Interlacing of the nested principal
submatrices gives the two eigenvalue inequalities and preserves $\eta\mu_{\max}\le1$.
Applying the same threshold monotonicity proves the placement comparison.
\end{proof}

\begin{figure}[t]
  \centering
  \includegraphics[width=0.72\linewidth]{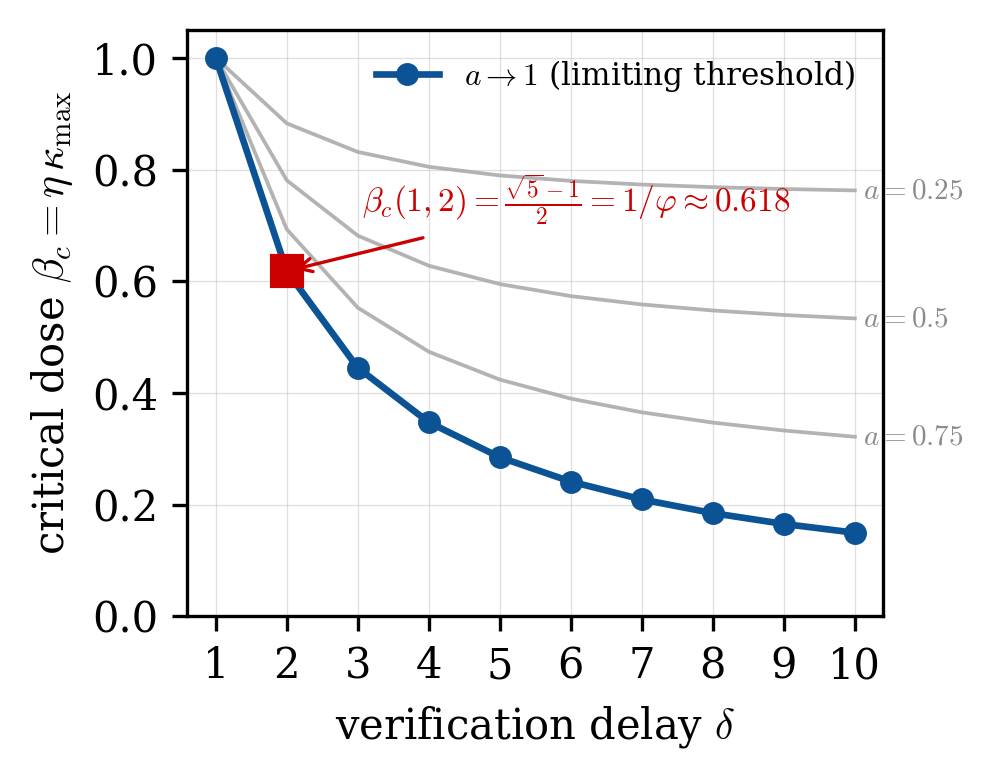}
  \caption{\textbf{The verification dose ceiling falls with delay.} Critical dose
  $\beta_c=\eta\kappa_{\max}$ versus verification delay $\delta$ for the limit $a\uparrow1$ (blue) and
  three positive modes (grey); each scalar mode is stable below its threshold. The ceiling decreases
  with $\delta$ for these positive values of $a$; at $\delta=2$ the limiting blue curve equals
  the inverse golden ratio $(\sqrt5-1)/2\approx0.618$ (red).
  More verification \emph{latency} therefore forces a strictly smaller safe verification
  \emph{strength}: the dose--delay tradeoff that the rest of the paper exploits.}
  \label{fig:dose}
\end{figure}

\section{Corrector placement for a coherence objective}\label{sec:place}

Given a budget of $k$ correctors, \emph{where} should they go? This section first bounds the
equilibrium error for the hard-grounded model, then defines a separate soft-pin placement
problem on a fixed state space. We optimize a coherence surrogate,
whose reduction is submodular. The greedy rule of Algorithm~\ref{alg:greedy} achieves a fraction
$1-1/e$ of the optimal reduction (Theorem~\ref{thm:place}). This is a guarantee for the chosen
placement metric; it is not a guarantee for every measure of truth-tracking error.

The delay drops out of the equilibrium of \eqref{eq:dyn}:
\begin{equation}\label{eq:ss}
  (\Lg+\kappa I)\,e_\infty = g \quad\Longrightarrow\quad e_\infty(R)=(\Lg+\kappa I)^{-1}g,
\end{equation}
so \emph{stability} is governed by $(\spec\Lg,\kappa,\delta)$ while \emph{truth-tracking error} is
governed by the resolvent $(\Lg+\kappa I)^{-1}$ acting on the fault forcing. Placement $R$ affects
both quantities; these are distinct criteria rather than independently adjustable levers.

\begin{corollary}[bounded steady-state error]\label{cor:bibo}
Whenever the loop is stable, $e_t\to e_\infty$ and the truth-tracking error obeys the resolvent bound
\begin{equation}\label{eq:bibo}
  \|e_\infty\|_2 \;\le\; \frac{\|g\|_2}{\mu_{\min}(\Lg)+\kappa},
\end{equation}
independent of the step size $\eta$ and the delay $\delta$. Thus the same $\mu_{\min}(\Lg)$ that the
dose limit (Corollary~\ref{cor:dose-explicit}) makes the binding stability constraint also controls the
residual error. Adding hard correctors never decreases $\mu_{\min}(\Lg)$ (Cauchy interlacing),
so the bound does not increase if the forcing on the remaining nodes is unchanged. Under the
step-size condition of Corollary~\ref{cor:dose-explicit}, the dose ceiling also does not decrease.
\end{corollary}
\begin{proof}
At equilibrium \eqref{eq:dyn} gives $(\Lg+\kappa I)e_\infty=g$; since $\Lg\succ0$ is symmetric,
$\|(\Lg+\kappa I)^{-1}\|_2=1/(\mu_{\min}(\Lg)+\kappa)$, and submultiplicativity gives \eqref{eq:bibo}.
Subtracting $e_\infty$ gives the stable homogeneous recurrence, proving convergence for every
initial history. Restricting an unchanged forcing to fewer free nodes cannot increase its norm.
\end{proof}

The bound \eqref{eq:bibo} shows that pinning helps, but not \emph{which} nodes to pin: it sees only
$\mu_{\min}(\Lg)$. For the placement objective, fix a baseline hard-corrector set $R_0$ and
$G=L[V\setminus R_0,V\setminus R_0]$. The set $R_0$ may be empty, as in the numerical example;
$\kappa>0$ still makes $G+\kappa I$ positive definite. In the rest of this section, $R$ denotes
additional \emph{soft} pins on this fixed state space, each applying a finite restoring term
$-w e_i$ in the steady-state equation rather than imposing $e_i=0$. We write $u_i$ for the
standard coordinate vector at node $i$, to distinguish it from a node's error. Their steady-state
operator is
\begin{equation}\label{eq:Mop}
  M(R)=G+W_R+\kappa I,\qquad W_R=w\sum_{i\in R}u_iu_i^\top .
\end{equation}
If $g$ is sampled once, held constant in time, and has mean zero and covariance $\sigma^2 I$, then
\[
  \mathbb E\|M(R)^{-1}g\|_2^2=\sigma^2\tr M(R)^{-2}.
\]
This differs from the \emph{coherence} $H(R)=\tr M(R)^{-1}$ used in leader-selection objectives
\cite{ClarkSubmod,MackinPatterson}. For the separate, continuous-time, delay-free surrogate
$\mathrm dx=-M(R)x\,\mathrm dt+\mathrm dW_t$, where $W_t$ is standard vector Brownian motion
with independent unit-variance components, the stationary covariance $P$ solves
$M P+P M=I$. Hence $P=\tfrac12 M^{-1}$ and the stationary energy $\mathbb E\|x\|_2^2$ is $H(R)/2$.
This noise model is not the constant forcing in \eqref{eq:dyn}, nor does it give the stochastic
variance of the delayed discrete-time loop.

Both traces are non-increasing under added pins: each ordered eigenvalue of $M$ is
non-decreasing under the positive-semidefinite update, so both $\sum_j\lambda_j(M)^{-1}$
and $\sum_j\lambda_j(M)^{-2}$ are non-increasing. They need not rank two non-nested placements alike.
For example, take the six-node graph with edges
$\{02,03,04,05,12,25,34,35\}$ (unit weights; nodes $0,\ldots,5$, with $ij$ denoting
$\{i,j\}$), $R_0=\varnothing$, $\kappa=1/2$, and $w=2$:
\begin{center}
\begin{tabular}{c|cc}
Placement $R$ & $H(R)=\tr M^{-1}$ & $\tr M^{-2}$\\\hline
$\{0,1\}$ & $2.440088$ & $1.640781$\\
$\{1,4\}$ & $2.383134$ & $1.666253$
\end{tabular}
\end{center}
The second placement improves coherence but worsens the static mean-square error. We therefore
state the optimization problem explicitly for the coherence surrogate:
\begin{equation}\label{eq:placeprob}
  \min_{|R|=k}\ H(R)=\tr M(R)^{-1},
\end{equation}
where every node in the fixed set $V\setminus R_0$ is eligible and $1\le k\le |V\setminus R_0|$.
The optimization is over this static objective; it does not impose a delayed-loop stability
constraint. Implementing the selected soft pins in a dynamic controller requires a separate
stability check for that controller.

The coherence is supermodular, so its reduction $\rho(R)=H(\varnothing)-H(R)$ is monotone and submodular,
and greedy achieves at least a fraction $1-1/e$ of the optimal \emph{reduction} \cite{Nemhauser}
(Theorem~\ref{thm:place}).
The greedy step is cheap and interpretable, because its \emph{marginal gain} has the Sherman--Morrison
closed form
\begin{equation}\label{eq:marg}
  \Delta_i(R):=H(R)-H(R\cup\{i\})=\frac{w\,\|M(R)^{-1}u_i\|^2}{1+w\,u_i^\top M(R)^{-1}u_i}\;\ge\;0,
\end{equation}
a \emph{resolvent centrality} that measures the coherence reduction from one additional pin and
depends on the full inverse, rather than node degree alone. Greedy pins the node of largest $\Delta_i$, updates
$M(R)^{-1}$ by one rank-one step, and repeats (Algorithm~\ref{alg:greedy}).
Here \emph{submodular reduction} means diminishing returns:
$\Delta_i(R)\ge\Delta_i(S)$ whenever $R\subseteq S$ and $i\notin S$.
Equivalently, the decreasing cost $H$ is supermodular.

\begin{algorithm}[H]
\DontPrintSemicolon
\caption{Greedy corrector placement for problem~\eqref{eq:placeprob}: spend the budget one node at a
time, each step pinning the node of largest marginal gain.}\label{alg:greedy}
\KwIn{fixed baseline Laplacian $G$, dose $\kappa>0$, pin weight $w>0$, budget $k$}
\KwOut{corrector set $R$ with $|R|=k$}
$R \leftarrow \varnothing$\;
\For{$j \leftarrow 1$ \KwTo $k$}{
  \lForEach{$i \notin R$}{$\Delta_i(R) \leftarrow w\,\|M(R)^{-1}u_i\|^2 \big/ \bigl(1+w\,u_i^\top M(R)^{-1}u_i\bigr)$ \tcp*[r]{\eqref{eq:marg}}}
  $i^\star \leftarrow \arg\max_{i \notin R}\,\Delta_i(R)$ \tcp*[r]{largest coherence reduction}
  $R \leftarrow R \cup \{i^\star\}$\;
}
\Return $R$\;
\end{algorithm}

\begin{theorem}[greedy coherence-reduction guarantee]\label{thm:place}
Let $G$ be a fixed symmetric graph Laplacian or a principal submatrix of one, and let $\kappa,w>0$.
The coherence $H(R)=\tr M(R)^{-1}$ is non-increasing and supermodular. Its reduction
$\rho(R)=H(\varnothing)-H(R)$ is non-decreasing and submodular with $\rho(\varnothing)=0$.
For each integer budget $1\le k\le |V\setminus R_0|$, the greedy rule satisfies
\begin{equation}\label{eq:greedybound}
  \rho(R_{\mathrm{greedy}})\;\ge\;\Bigl(1-\bigl(1-\tfrac1k\bigr)^{k}\Bigr)\,\rho(R^\star)\;\ge\;\bigl(1-\tfrac1e\bigr)\,\rho(R^\star),
\end{equation}
where $R^\star\in\arg\max_{|R|=k}\rho(R)$ is an optimal size-$k$ set. In terms of the minimized
objective, the last inequality is equivalent to
\[
  H(R_{\mathrm{greedy}})\le\frac{H(\varnothing)}{e}
       +\Bigl(1-\frac1e\Bigr)H(R^\star).
\]
This is not a multiplicative approximation factor for $H$, and gives no approximation factor
for $\tr M(R)^{-2}$.
\end{theorem}
\begin{proof}
See \ref{app:place}.
\end{proof}

On a graph of three $5$-cliques chained by unit-weight bridges, with $R_0=\varnothing$,
$\kappa=0.3$ and $w=12$ (Fig.~\ref{fig:placement}), the rule is well separated
from the two baselines at $k=8$: greedy cuts the coherence $\tr M(R)^{-1}$ from $9.1$ to
$2.3$, against $2.9$ for degree-based and $2.5$ for random placement. Its first four picks are
$4,13,8,0$ in the code's zero-based indexing. The supermodular structure and the $(1-1/e)$ guarantee are \emph{inherited}
from leader selection in linear multi-agent systems \cite{ClarkLeader,ClarkSubmod} and hold for the
coherence objective \eqref{eq:placeprob}; these comparisons do not measure the gap to an exact
optimum or establish a guarantee for a nonlinear error-propagation map. Hard grounding removes rows and columns: for nested corrector sets, the smallest
grounded eigenvalue cannot decrease and the largest cannot increase. One way to see both directions
is to extend a vector on the remaining nodes by zeros at the newly grounded nodes: its Rayleigh
quotient is unchanged, while the minimization and maximization domains are restricted. For a
three-node path, successively grounding its endpoints gives the spectra
$\{0,1,3\}$, $\{(3-\sqrt5)/2,(3+\sqrt5)/2\}$, and $\{2\}$.
This differs from a finite soft pin in \eqref{eq:Mop}, which adds a positive-semidefinite diagonal
matrix on a fixed state space. These spectral facts do not establish an interior optimal budget
$k^\star$: such a claim needs an explicit cost or another competing constraint, beyond the
fixed-budget placement problem studied here.

\begin{figure}[t]
  \centering
  \includegraphics[width=0.92\linewidth]{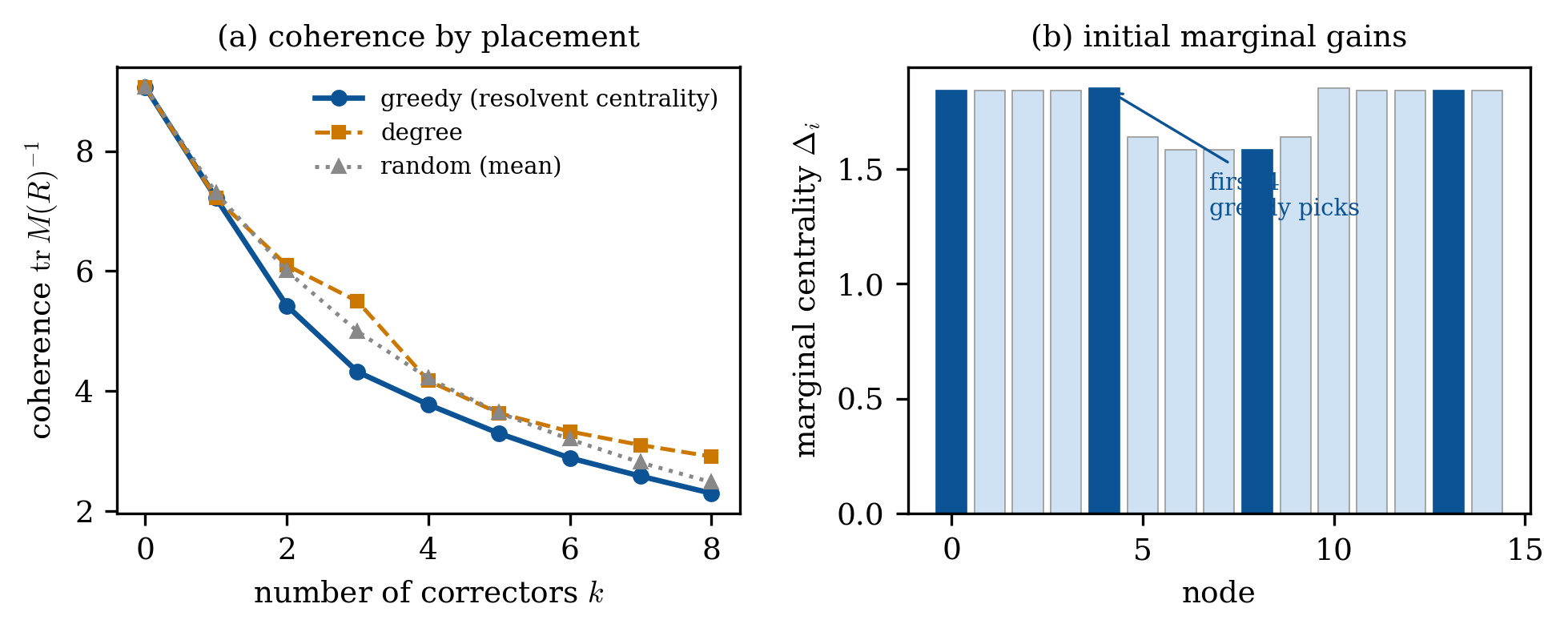}
  \caption{\textbf{Where to place correctors.} (a)~On three $5$-cliques chained by bridge edges, greedy
  selection by the resolvent centrality \eqref{eq:marg} is compared with degree-based and random
  placement (mean over $300$ orders) using coherence $\tr M(R)^{-1}$. An exact optimum is not plotted.
  (b)~Initial marginal gains $\Delta_i(\varnothing)$ per node; the first four sequential greedy picks
  are marked in dark blue. Later picks use gains recomputed after the preceding pins.}
  \label{fig:placement}
\end{figure}

\section{Two coupled delays}\label{sec:two}

Let gossip carry integer latency $d\ge0$ and verification integer latency $\delta\ge1$:
\[
  e_{t+1}=e_t-\eta\Lg\,e_{t-d}-\eta\kappa\,e_{t-\delta}+\eta g .
\]
Lemma~\ref{lem:decouple} still applies, giving per mode
$x_{t+1}=x_t-\eta\mu\,x_{t-d}-\eta\kappa\,x_{t-\delta}$. A unit-circle root
$\lambda=e^{i\theta}$ satisfies the pair
\begin{equation}\label{eq:twodelay}
  \cos\theta-1+\eta\mu\cos(d\theta)+\eta\kappa\cos(\delta\theta)=0,\quad
  \sin\theta-\eta\mu\sin(d\theta)-\eta\kappa\sin(\delta\theta)=0.
\end{equation}
With $h=\max(d,\delta)$ and a supplied history $e_{-h},\ldots,e_0$, the characteristic polynomial is
\[
 P(z)=z^{h+1}-z^h+pz^{h-d}+qz^{h-\delta},\qquad p=\eta\mu,\quad q=\eta\kappa.
\]
It generally has four terms when $d,\delta>0$ are distinct. Related trinomial delay equations are
studied in \cite{KipnisNigmatullin}; we derive the unit-circle locus for the polynomial above directly.

\begin{theorem}[two-delay unit-circle locus]\label{thm:twodelay}
For $d\ne\delta$ and $\theta\in(0,\pi)$ with $\sin((\delta-d)\theta)\ne0$, a root
$z=e^{i\theta}$ occurs exactly at
\[
  p(\theta)=\frac{\sin(\delta\theta)-\sin((\delta+1)\theta)}{\sin((\delta-d)\theta)},\qquad
  q(\theta)=\frac{\sin((d+1)\theta)-\sin(d\theta)}{\sin((\delta-d)\theta)}.
\]
At angles with $\sin((\delta-d)\theta)=0$, the original two linear equations
\eqref{eq:twodelay} must instead be checked for consistency; when consistent, they give additional
lines in the $(p,q)$ plane. The real-root loci are $p+q=0$ for $z=1$ and
$p(-1)^d+q(-1)^\delta=2$ for $z=-1$. These sets comprise the unit-circle locus, not a selection
of its stability-bounding arcs. In each connected open cell of its complement, the number of roots
inside the unit disk is constant; stable cells are those with all $h+1$ roots inside.
The origin itself is marginal, with a root at $z=1$.
\end{theorem}
\begin{proof}
See \ref{app:twodelay}.
\end{proof}

For example, at $(d,\delta)=(1,4)$ and $\theta=\pi/3$, the denominator vanishes and
\eqref{eq:twodelay} reduces to $p-q=1$. On that line,
\[
 P(z)=(z^2-z+1)(z^3+qz+q).
\]
At $q=1/10$, the quadratic roots lie on the unit circle while the cubic roots lie strictly inside:
on $|z|=1$, $|qz+q|\le1/5<|z^3|$, so Rouch\'e's theorem applies. Thus singular cases cannot be
discarded when constructing the boundary.

\begin{corollary}[equal-delay combined-gain ceiling]\label{cor:equal}
If $d=\delta=\tau\ge1$, the two terms merge into a single lag at $a=1$,
$x_{t+1}-x_t+\eta(\mu+\kappa)x_{t-\tau}=0$. For positive combined gain, it is stable iff
\[
  0<\eta(\mu+\kappa)<\beta_c(1,\tau)=2\sin\frac{\pi}{4\tau+2},
  \qquad\text{e.g.}\quad \beta_c(1,2)=\frac{\sqrt5-1}{2}=\frac1\varphi .
\]
For the network, this becomes $\eta(\mu_{\max}(\Lg)+\kappa)<\beta_c(1,\tau)$: the largest
grounded eigenvalue binds, unlike the instantaneous-gossip case of Corollary~\ref{cor:dose-explicit}.
Here $\varphi=(1+\sqrt5)/2$ is the golden ratio.
\end{corollary}
\begin{proof}
See \ref{app:twodelay}.
\end{proof}

This is an exact special case, not a universal ordering of delay configurations. Monotonicity in
$a=1-\eta\mu$ for instantaneous gossip does not compare different positive gossip delays. In
particular, at $\tau=2$ the golden-ratio threshold applies to the \emph{combined} gain $p+q$.
For instantaneous gossip, the verification-only threshold $\beta_c(1-p,2)$ approaches $1/\varphi$
from above as $p\downarrow0$.

A fixed-total-delay example shows why a comparison class is necessary. Fix $p=1/10$ and
$q=98/125$. At $(d,\delta)=(1,1)$ the polynomial is $z^2-z+221/250$; both roots have modulus
$\sqrt{221/250}<1$. With the same gains and total delay $d+\delta=2$, the allocation $(0,2)$ gives
\[
 z^3-\tfrac9{10}z^2+\tfrac{98}{125}
   =(z+\tfrac7{10})(z^2-\tfrac85 z+\tfrac{28}{25}).
\]
Its conjugate roots $(4\pm2i\sqrt3)/5$ have modulus $\sqrt{28/25}>1$.
Thus equal delays are stable here while unequal delays are unstable. This example compares a fixed
sum of delays; it does not establish an ordering for a fixed maximum delay or a fixed verification delay.

\begin{figure}[t]
  \centering
  \includegraphics[width=0.7\linewidth]{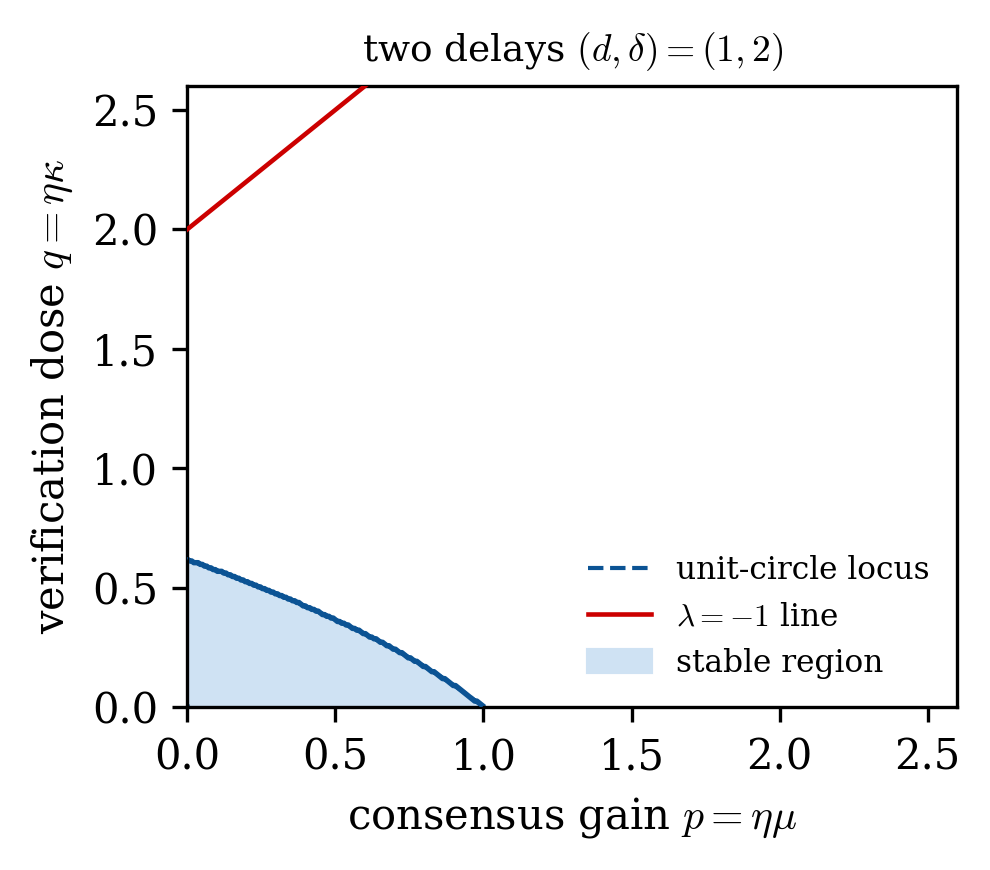}
  \caption{\textbf{Stability for a specified pair of delays.} Stability region (shaded) in the
  $(p,q)=(\eta\mu,\eta\kappa)$ plane for communication delay $d=1$ and verification delay $\delta=2$,
  computed by testing all characteristic roots. The regular unit-circle locus is dashed
  (Theorem~\ref{thm:twodelay}); the $\lambda=-1$ line is red and here non-binding.
  The dose ceiling on the $q$-axis is $1/\varphi\approx0.618$. This plot makes no general comparison
  across delay allocations.}
  \label{fig:twodelay}
\end{figure}

\section{Experiments and empirical scope}\label{sec:emp}

We ask three questions, with an archival and an expanded study for RQ2. RQ1: does the nonlinear loop lose stability
near the predicted dose limit $\kappa_{\max}(\delta)$? RQ2: in a real grounded factual debate,
how do final-step error and finite-window variability differ across verifier prompts and delays?
RQ3: do LLM agents transmit an externally computed delayed correction signal, and how does their
finite-window signed-error variability change with its lag?

The studies use different outcomes. RQ1 measures the late half peak-to-peak range of a simulated
real-valued error. Historical RQ2 uses an automated entailment score; expanded RQ2 separates
binary factual labels from abstentions and unresolved answers. RQ3 measures the temporal
standard deviation of a normalized signed numeric error. We retain the original name
``amplitude'' where used in the records, but give its definition for each study; values are
not directly comparable across these outcomes. Unless stated otherwise, temporal standard
deviations are population standard deviations, dividing by the number of observed states.

\subsection{Onset at the predicted dose limit (RQ1)}

We test whether the \emph{linear} dose limit predicts onset in the \emph{nonlinear} system with a
saturating verifier,
\[
  e_{t+1}=(I-\eta\Lg)e_t-\eta\kappa\tanh(e_{t-\delta})+\eta g
\]
($\tanh$ acts componentwise and $\tanh'(0)=1$ matches the zero-state linearization).
The ten-node graph has two hard pins (nodes $0,5$ in zero-based code indexing), leaving
$n_{\mathrm f}=8$; it uses a random spanning tree with additional independent edges of probability
$0.3$ and seed $20260621$. We set $\eta=0.5/\mu_{\max}$ and a forcing $0.04$ at the second
free node, zero elsewhere, starting from an all-zero history. Each run has $4000$ updates.
Its late amplitude is half the largest nodewise peak-to-peak range over the final $400$ states.
The observed onset
$\kappa_{\mathrm{crit}}$ tracks the predicted ceiling
$\kappa_{\max}=\beta_c(1-\eta\mu_{\min},\delta)/\eta$:
\begin{center}
\begin{tabular}{cccc}
\hline
$\delta$ & $\kappa_{\max}$ (theory) & $\kappa_{\mathrm{crit}}$ (nonlinear) & ratio\\
\hline
1 & 16.30 & 16.57 & 1.017\\
2 & 10.21 & 10.38 & 1.017\\
3 & \phantom{0}7.45 & \phantom{0}7.57 & 1.017\\
\hline
\end{tabular}
\end{center}
With nonzero forcing, zero need not be an equilibrium of the nonlinear loop; the comparison uses
the zero-state linear reference and is not an exact local stability calculation at the forced equilibrium.
The reported $\kappa_{\mathrm{crit}}$ is the first sampled dose whose late amplitude exceeds $10^{-3}$.
The sweep uses $40$ equally spaced ratios $\kappa/\kappa_{\max}$ from $0.35$ to $1.65$: its spacing is
$0.0333$, and the samples bracketing the predicted boundary are $0.9833$ and $1.0167$. Thus the common
ratio $1.017$ is consistent with onset at the theoretical boundary at this resolution; it does not
resolve a separate stabilizing shift due to saturation. At $\delta=2$ and
$\kappa=1.25\kappa_{\max}$, the dominant nonzero Fourier bin of the first free node's
final $600$ states gives a nominal period $9.38$, compared with $2\pi/\theta_\star=9.72$
at the linear boundary. This finite-window frequency estimate is not an exact period or a
measurement at the threshold. An LLM front-end requires its own identification and cannot
be assumed to preserve this update map.
\begin{figure}[t]
  \centering
  \includegraphics[width=0.9\linewidth]{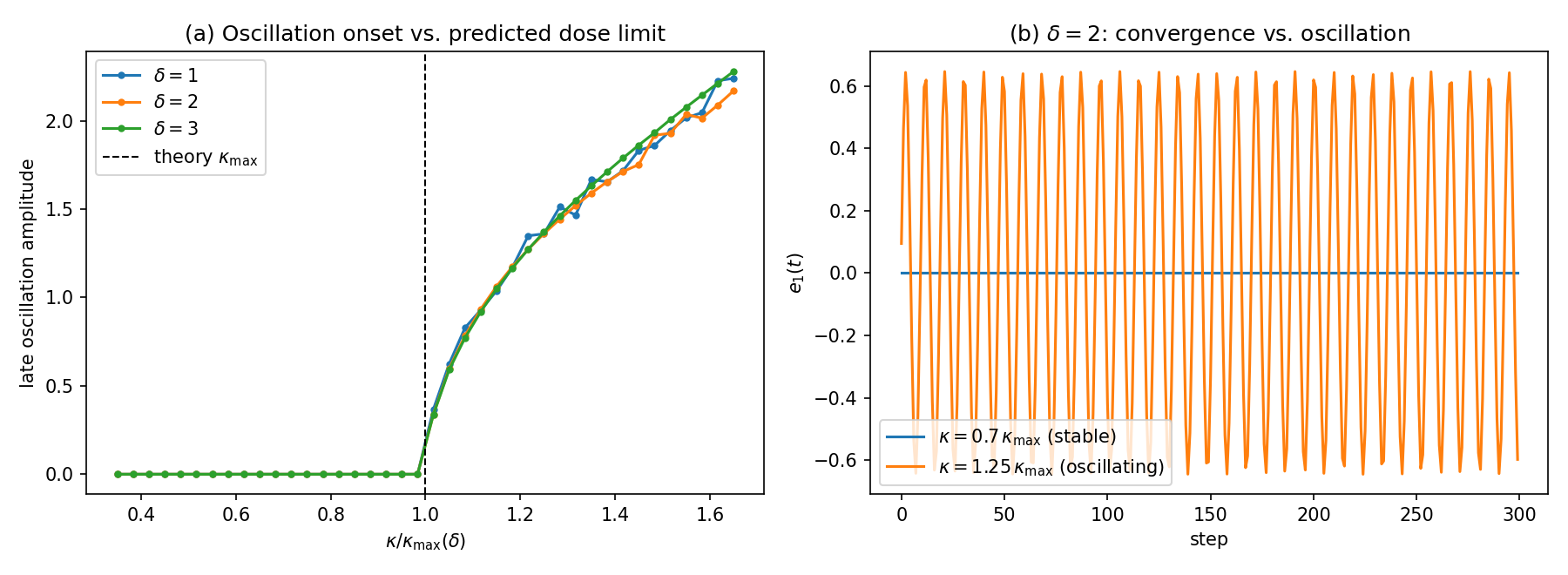}
  \caption{\textbf{Synthetic onset is consistent with the predicted ceiling at the sweep resolution.} (a)~oscillation amplitude collapses
  onto the predicted threshold $\kappa/\kappa_{\max}=1$ for $\delta=1,2,3$; (b)~the $\delta{=}2$
  trajectory converges below the ceiling and oscillates above it.}
  \label{fig:synth}
\end{figure}

\subsection{Grounded factual debate: protocol-dependent outcomes (RQ2)}

We reanalyze archived debates among instances of Qwen3.6-35B on selected PsiloQA
\cite{PsiloQA} and TruthfulQA \cite{TruthfulQA} questions. The archived filename suffix \texttt{tqa} refers to TruthfulQA: all $30$ selected
questions and gold answers match that dataset; the earlier TriviaQA label was erroneous.
Selection used high cold-answer natural-language-inference (NLI) error
within a candidate set, so these are difficulty-selected subsets rather than representative samples.
Three free agents receive peer answers and a verifier note; grounded variants also supply four fixed
wrong answers. Evidence is a Wikipedia passage for PsiloQA and a supplied gold-answer/alias statement
for TruthfulQA. The verifier sees the first $600$ characters of that evidence and is itself an LLM,
not an oracle. The score $(1-P(\mathrm{entailment})+P(\mathrm{contradiction}))/2$ from
\texttt{microsoft/deberta-large-mnli}\footnote{Model card:
\url{https://huggingface.co/microsoft/deberta-large-mnli}.}
is a bounded error surrogate; no human calibration of this score is reported.

The $30$-question PsiloQA sweep has $120$ runs ($18$ states; weak/strong prompts and two stored
delay labels). A separate $20$-question sweep has $80$ runs ($36$ states; strong/forceful prompts).
In the latter, the final-step mean score is below $0.1$ in $16/20$ questions with the strong prompt
and $8/20$ with the forceful prompt, at either delay. This is an \emph{endpoint success fraction},
not a convergence test. On the analogous TruthfulQA subset, the corresponding fractions are $4/20$
for the strong prompt and $3/20$ or $2/20$ for the forceful prompt. In the separate ungrounded
contrarian protocol, endpoint success is zero in all four cells. The later $80$-run archive with
saved majority strings has a mean normalized-string change rate of $90.8\%$ after two warmup states.
String changes need not be semantic changes, and the contrarian protocol also changes prompts and
removes the faulty majority; it is not an isolated grounding ablation.

The F1 archive contains $180$ runs per dataset: $15$ questions, three delay labels, two forcing
levels, two update temperatures, one repetition, and $20$ states. The written plan instead specified
five delays, three forcing levels, $30$ questions, three repetitions, $24$ states and a verifier-off
control. Those missing conditions were not executed in the archived sweep. The analysis also
substituted a two-delay Wilcoxon contrast for the planned ordered trend test. We therefore treat F1
as \emph{exploratory}, not as execution of its confirmatory decision rule.

There is an additional indexing defect. For an update producing state $t$, the current state is
$t-1$; the archived F1 code uses the current state for label $r=0$ and state $\max(0,t-r)$ otherwise.
Labels $r=0$ and $r=1$ therefore both implement lag zero, and $r=6$ implements lag five. The other
factual harnesses similarly map labels $1,4$ to lags $0,3$, and labels $1,6$ to lags $0,5$.
The corrected code records and uses theoretical indexing by default; these archived data remain
unchanged. Furthermore, the nominal temperature-zero F1 condition used temperature $0.7$ for initial
answers, so it was not deterministic from initialization.

We recompute amplitude as the standard deviation of the final $10$ mean-score states and pair
contrasts by question. With four faulty peers, the label-$6$ minus label-$0$ mean differences and
unadjusted one-sided Wilcoxon $p$-values are:
\begin{center}
\begin{tabular}{lccc}
\hline
Dataset & Update temperature & Mean amplitude difference & $p$\\
\hline
PsiloQA & $0$ & $-0.0044$ & $0.704$\\
PsiloQA & $0.7$ & $+0.0023$ & $0.433$\\
TruthfulQA & $0$ & $+0.0317$ & $0.124$\\
TruthfulQA & $0.7$ & $+0.0350$ & $0.094$\\
\hline
\end{tabular}
\end{center}
Each contrast has $15$ question pairs. These tests do not reject at $0.05$, even before correction
for multiple comparisons; that is not evidence of equivalence or stability. Mean amplitude increases
in three of the four contrasts. Without faulty peers, median amplitudes are often zero, but some
trajectories still vary (maximum tail standard deviation $0.478$). Variability at lag zero shows that
delay is not necessary for variability in this protocol; it does not identify a unique cause.
The archives omit per-agent answer text and verifier responses, preventing independent NLI rescoring.
Non-negativity alone does not imply delay stability (Remark~\ref{rem:absorb}).

\subsubsection{Expanded factual study with complete response logs}\label{sec:expanded}

To address the small historical samples and missing response text, we ran a separate study on
September 13--15, 2026 using the server model tag \texttt{qwen3.8:latest} through native Ollama.
The frozen manifest records its artifact digest; the server tag is an operational identifier,
not independent verification of a model-family or parameter-count claim. This study contains
$200$ new PsiloQA questions and $200$ new TruthfulQA questions, excluding previously used questions.
Selection used a deterministic hash ordering and source review before the new trajectories,
without screening on new model responses or an NLI difficulty threshold. PsiloQA selection was
from an exported pool of $4000$ candidates, with at most one question per source article;
TruthfulQA selection used the available unseen pool and an $80$-character reference limit.
These restrictions do not yield representative samples of all factual questions. The locally
frozen protocol followed earlier pilots; it was not a public preregistration, and the earlier
NLI-based power estimate does not transfer to this changed outcome measure.

Three free agents interact with four fixed erroneous peers. For each question, three initial
answers are generated once and shared between the two conditions, with actual verification delays $\delta=0,5$.
Each trajectory has $48$ states, indexed $0,\ldots,47$; an update producing state $t$ supplies the
verifier with state $\max(0,t-1-\delta)$. Here $\delta=0$ denotes immediate verification; communication has no added lag.
Agent updates are synchronous, and the temperature is $0.7$
for initialization and updates, with a temperature-zero verifier. Each request has a recorded seed;
condition order is fixed by a hash. The verifier is another invocation of the same model, using
up to $600$ characters of frozen evidence. TruthfulQA evidence includes the reference answer:
this is delayed reference feedback, not independent evidence retrieval. The primary grid holds
the verifier on and the number of erroneous peers fixed; verifier-off and additional forcing
conditions were not part of this expanded grid.

The output contract permits a short answer of at most $80$ characters or an explicit abstention,
with context $4096$, generation budget $256$ tokens and thinking disabled. This differs from the
historical free-text verifier note, so the new study is not an isolated intervention on indexing.
All $800$ trajectories and $226800$ requests completed; request-level parsing yielded $200065$
valid answers, $26733$ abstentions and two invalid answers, with no transport errors or
length-limit terminations. Valid formatting does not ensure a complete or true statement.
The request count includes verifier calls and differs from the number of stored agent states.
Full logs reproduce prompts, seeds, synchronous states and delay indices; source-file hashes
match the technical replay. Pilots and historical repetitions are excluded from the $400$ questions.

\paragraph{Versioned scoring.}
The initial scorer uses a question-specific frozen catalogue, assigning M (reference match),
C (reference contradiction), U (unresolved answer), A (abstention) or F (invalid format).
Only M/C receive binary values $0/1$. Of the unique pairs initially labelled U, all $3604$ were
read by the assisting language model in a packet hiding delay, trajectory position and frequency.
The same assistant had already seen aggregate results, so this is neither independent human
adjudication nor fully blind review. A source-conservative first pass, v1, assigned
$1915$ M, $662$ C and $1027$ U. A second pass, v2, assessed all remaining $1027$ pairs by meaning,
consulting external sources for selected disputed facts; it added $734$ M and $113$ C, leaving
$180$ U. Identical question--answer pairs retain identical labels in every occurrence.

Version v2 accepts recognizable spelling variants, some incomplete endings with unambiguous main
content, and explicitly qualified descriptions of folklore. These post-hoc interpretive choices
are not the official TruthfulQA evaluation, and the assistant's judgments are not certified ground
truth. Original M/C labels, abstentions and format failures were not relabelled. Both versions,
their reasons and sources are preserved; decisions were frozen before their respective effect
recalculations. Together the reviews classify $50149$ of the $51670$ initially unresolved stored
agent states ($97.1\%$), leaving $1521$ U. This percentage measures scoring coverage, not accuracy.

\paragraph{Binary-amplitude contrast.}
For one question, let $y_{i,t,\delta}=0$ for label M and $1$ for label C; this binary
outcome is distinct from the real-valued belief error $e_{i,t}$ in the model.
It is undefined for A, U and F. For a fully scored or hypothetically completed trajectory let
\[
 A_\delta=\operatorname{SD}_{t=24,\ldots,47}
       \left(\frac{1}{3}\sum_{i=1}^{3} y_{i,t,\delta}\right),
 \qquad \Delta=A_5-A_0,
\]
using population standard deviation. We write $\operatorname{SD}_{t\in T}(v_t)
=[|T|^{-1}\sum_{t\in T}(v_t-\bar v_T)^2]^{1/2}$, where
$\bar v_T=|T|^{-1}\sum_{t\in T}v_t$. For each question we compute the minimum and maximum
$A_\delta$ over binary completions of unscored cells, then the outer interval
$[A_{5,\min}-A_{0,\max}, A_{5,\max}-A_{0,\min}]$. Table~\ref{tab:expanded-bounds}
averages these endpoints over all $200$ questions per dataset. For absent answers these
completions are hypothetical: an abstention has no hidden factual truth value. Independently
completing residual U occurrences can also be more permissive than enforcing one truth label
across all occurrences of the same pair. Thus these are conservative outer bounds, not sharp
semantic identification regions, confidence intervals or significance tests. They include neither
sampling uncertainty nor possible adjudicator error. The finite optimization can be reproduced
without enumerating all cell assignments: at each tail time, the possible erroneous-agent counts
are consecutive integers from the known C count to that count plus the unscored count.
A dynamic program starts at $(S,Q)=(0,0)$ and, for each possible count $c$ at the next time,
adds $(S+c,Q+c^2)$. It therefore retains exactly the reachable sums of counts and squared counts.
For $N=24$ tail times and three agents, the variance is $Q/(9N)-S^2/(9N^2)$; taking the
square roots of its minimum and maximum over reachable pairs gives the completion extrema.
The implementation keeps only the smallest and largest $Q$ for each $S$: future count choices
do not depend on $Q$, and the final variance is increasing in $Q$ at fixed $S$.

\begin{table}[t]
\centering
\small
\begin{tabular}{lccc}
\hline
Dataset & Catalogue only & Semantic v1 & Semantic v2\\
\hline
PsiloQA & $[-0.31336,\,0.30800]$ & $[-0.21964,\,0.22103]$ & $[-0.21110,\,0.21280]$\\
TruthfulQA & $[-0.41625,\,0.42362]$ & $[-0.18397,\,0.17875]$ & $[-0.14170,\,0.13446]$\\
\hline
\end{tabular}
\caption{\textbf{Expanded factual study: conservative completion bounds for mean $\Delta$.}
All $200$ questions per dataset are retained. Both signs remain possible under each scoring
version. These intervals are not confidence intervals and do not establish absence of an effect.}
\label{tab:expanded-bounds}
\end{table}

In the analysis tail, PsiloQA has $7781$ A and $292$ U cells, and TruthfulQA has $3510$ A,
$417$ U and two F cells, each out of $28800$ cells. Abstentions account for $11291/12002=94.1\%$
of cells without a binary score, not necessarily the same fraction of interval width. Fully
scored tails in both conditions occur for only three PsiloQA and $15$ TruthfulQA questions;
restricting inference to those questions would select on model behavior. The direction of the
original binary-amplitude contrast remains unidentified after the larger collection and review.

\paragraph{Separate analysis of observed outcomes.}
To describe behavior without assigning factual truth to an abstention, we subsequently defined
$p_{k,t,\delta}$ as the fraction of the three agents receiving label $k\in\{\mathrm M,\mathrm C,
\mathrm A,\mathrm U,\mathrm F\}$. For each question we calculate its tail mean and population
standard deviation, retaining every label and all positions in the denominator. In particular,
$\operatorname{SD}_t(p_{\mathrm C,t,\delta})$ measures fluctuations in \emph{labelled erroneous
responses}; it is a different estimand from $A_\delta$ with undefined truth labels. Constant abstention
can have frequency one and standard deviation zero.

For the mean paired difference of each outcome statistic, $20000$ bootstrap resamples select
whole, previously frozen question clusters with replacement, preserving both conditions and all
questions within each selected cluster. Each resample draws the original number of clusters
with equal probabilities; repeated clusters contribute repeated copies of their questions.
The resampled statistic remains a mean over included questions, whose count can vary when
cluster sizes differ. The reported interval is the empirical $2.5$th to $97.5$th percentile.
There are $200$ PsiloQA source-article clusters and $151$ TruthfulQA question/family clusters,
the latter with maximum size $11$; resampling individual agent states would create
pseudoreplication. We use the pairs cluster-bootstrap principle \cite{CameronMiller} and report
pointwise percentile intervals. This is an exploratory analysis specified after reviewing earlier
results, without multiplicity adjustment or a confirmatory hypothesis test. The intervals condition
on the frozen labels and observed generation trajectories, assume independence between clusters,
and do not account for adjudicator error or variation over new model-generation seeds.

\begin{table}[t]
\centering
\small
\begin{tabular}{llrrr}
\hline
Dataset & Outcome frequency & $\delta=0$ (\%) & $\delta=5$ (\%) & Difference [95\% interval], pp\\
\hline
PsiloQA & M & $59.750$ & $61.215$ & $+1.465\;[-0.597,\,3.660]$\\
 & C & $11.840$ & $11.132$ & $-0.708\;[-2.389,\,0.792]$\\
 & A & $27.361$ & $26.674$ & $-0.687\;[-2.076,\,0.569]$\\
 & U & $1.049$ & $0.979$ & $-0.069\;[-0.271,\,0.076]$\\
 & F & $0.000$ & $0.000$ & $0.000\;[0.000,\,0.000]$\\
TruthfulQA & M & $78.090$ & $79.062$ & $+0.972\;[-0.375,\,2.162]$\\
 & C & $7.750$ & $7.812$ & $+0.062\;[-0.685,\,0.845]$\\
 & A & $12.799$ & $11.576$ & $-1.222\;[-1.910,\,-0.550]$\\
 & U & $1.347$ & $1.549$ & $+0.201\;[-0.618,\,1.110]$\\
 & F & $0.014$ & $0.000$ & $-0.014\;[-0.035,\,0.000]$\\
\hline
\end{tabular}
\caption{\textbf{Expanded factual study: post-hoc outcome frequencies under v2.}
Intervals concern paired changes in percentage points (pp), not separate condition means.
All positions, including abstentions, remain in the denominator. A zero bootstrap interval for
an unobserved event does not establish zero population probability.}
\label{tab:expanded-outcomes}
\end{table}

TruthfulQA abstention frequency is lower at lag five by $1.22$ percentage points, with a pointwise
$95\%$ interval $[-1.91,-0.55]$; this result is identical under v1 and v2
(Table~\ref{tab:expanded-outcomes}). Fewer abstentions do not by themselves establish more correct
answers. Indeed, the TruthfulQA M-frequency point difference changes from $-0.326$ percentage points
under v1 to $+0.972$ under v2, and both intervals include zero. For the standard deviation of
the labelled-error fraction, v2 gives changes $+0.00282$ ($95\%$ interval
$[-0.00771,0.01341]$) on PsiloQA and $+0.00211$ ($[-0.00485,0.00915]$) on TruthfulQA;
v1 likewise gives intervals spanning zero. These observations do not establish equivalence,
stability or the proposed feedback mechanism. Full five-outcome means, standard deviations,
paired results and both scoring versions are retained in the repository analysis files
(Appendix~\ref{app:repro}).

\subsection{Externally controlled signed-error variability across models (RQ3)}

We study a designed controller with LLM agents in its feedback loop, using eight author-curated
numeric questions with fixed reference values and positive normalization scales (in the same units
as each answer), recorded in \path{paper/remote/debate_v3.py}. The harness computes a suggestion
$V_{i,t}=b_{i,t}-\alpha(b_{i,t-\ell}-b^\star)$ outside the LLM, rounds it to one decimal place,
and supplies it alongside peer estimates. The prompt instructs agents to follow the suggestion and
peers and to use no outside knowledge. Thus the update law is imposed by the experimental protocol;
it is not independently identified from ordinary factual verification.

An indexing audit changes the mapping between the stored delay labels and the model. In the archived
harness, when producing state $t$, the current state is $t-1$ but the stale state is
$\max(0,t-r)$, where $r$ is the saved label. Hence exact following would implement
\[
 e_{t+1}=e_t-\alpha e_{t-\ell},\qquad \ell=r-1,
\]
with the initial state repeated before time zero. The saved labels $r=1,6$ therefore correspond to
lags $\ell=0,5$. For $\ell=0$, stability requires $0<\alpha<2$; for $\ell=5$, the ceiling is
$2\sin(\pi/22)\approx0.285$. The contrast at $\alpha=0.5$ still compares a stable and an unstable
ideal controller. But $\alpha=1.5,r=1$ is also stable in the ideal controller, with damped alternating
error; the claim that only the small-gain/small-delay cell is stable was incorrect.
The corrected harness uses the theoretical indexing by default and records the convention. No new
LLM runs under that corrected convention are reported here.

The numeric parser used in these archives extracts the last number-like substring; scientific
notation and leading-dot decimals can be misread. The updated harness uses a strict finite-number
parser. Raw responses were not archived, so the frequency of historical parsing errors cannot be
established. Scores and trajectories were stored to three decimal places: recomputed amplitudes
agree within the resulting $0.001$ tolerance, while two movement flags lie at a rounding-ambiguous
threshold. The analyses below retain the originally stored flags and amplitudes.

The stored experiment has $22$ states per run and three repeated runs per question and cell. Its
``amplitude'' is the standard deviation of $e_t=(\bar b_t-b^\star)/\mathrm{scale}$, where
$\bar b_t$ is the mean of the three free-agent estimates, over states
$3,\ldots,21$; it measures finite-window variability, including transients. The overshoot statistic
below requires both $\min_t e_t<-0.02$ and $\max_t e_t>0.02$ over the stored trajectory.
The separately stored \texttt{zc} field instead counts crossings of the trajectory's own centred
mean and is not used as a truth-overshoot statistic. Neither measure establishes persistent
oscillation, a nonlinear limit cycle, or a bifurcation.

Using all runs and averaging within question, the Qwen3.6-35B amplitude is $0.0549$ at $r=1$ and
$0.2648$ at $r=6$, a ratio of about $4.82$. All eight question-level differences are positive, giving
one-sided Wilcoxon $p=0.00390625$. The same all-run direction and $p$ occur for four other models;
these five survive a Bonferroni adjustment across the six tested models ($6p=0.0234375$).
Mistral-Nemo gives $p=0.0742$. These are exploratory all-run comparisons, not the
movement-conditioned analysis specified in the original script.

\begin{table}[t]
\centering
\caption{\textbf{All-run overshoot under the imposed controller.} Stored delay label $r$ corresponds
to effective lag $\ell=r-1$. Percentages use all $24$ trajectories per model and cell at
$\alpha=0.5$, with the two-sided truth-crossing threshold defined in the text.}\label{tab:models}
\begin{tabular}{llcc}
\hline
Model & Developer & $r=6$ ($\ell=5$) & $r=1$ ($\ell=0$)\\
\hline
Qwen3.6-35B & Qwen & $96\%$ & $0\%$\\
Qwen3-14B & Qwen & $100\%$ & $0\%$\\
Mistral-7B & Mistral & $100\%$ & $0\%$\\
Phi-4 & Microsoft & $96\%$ & $4\%$\\
Gemma-4-12B & Google & $100\%$ & $0\%$\\
Mistral-Nemo-12B & Mistral & $58\%$ & $0\%$\\
\hline
\end{tabular}
\end{table}

The specified movement filter retains runs with range greater than $0.15$ over states $3,\ldots,21$.
Applying it separately within each cell leaves $6$ paired questions for Qwen3.6
($p=0.015625$), $7$ for Gemma ($p=0.0078125$), $5$ for Phi-4 ($p=0.03125$), and $7$ for Nemo
($p=0.0546875$). Qwen3-14B and Mistral-7B have no retained short-lag runs, so that paired comparison
is undefined for them. The five-model all-run conclusion therefore cannot be presented as five
successful tests of the movement-conditioned plan. The filter itself selects on an outcome of the
intervention, so the conditional and unconditional comparisons answer different questions.

These results demonstrate that agents can transmit a delayed correction signal supplied by external
code, with model-dependent deviations from exact following. They do not establish that natural
verification has this update law, that deployed agents are necessarily more stable, or that the
finite observed trajectories approach a bounded limit cycle. The stored metadata also says three
faulty peers, whereas the available driver uses the function default of one; the executed count
cannot be resolved from the aggregate trajectories alone. The fields named ``seed'' are repetition
indices in this driver, which does not pass a random seed to the backend.

\begin{figure}[t]
  \centering
  \includegraphics[width=0.92\linewidth]{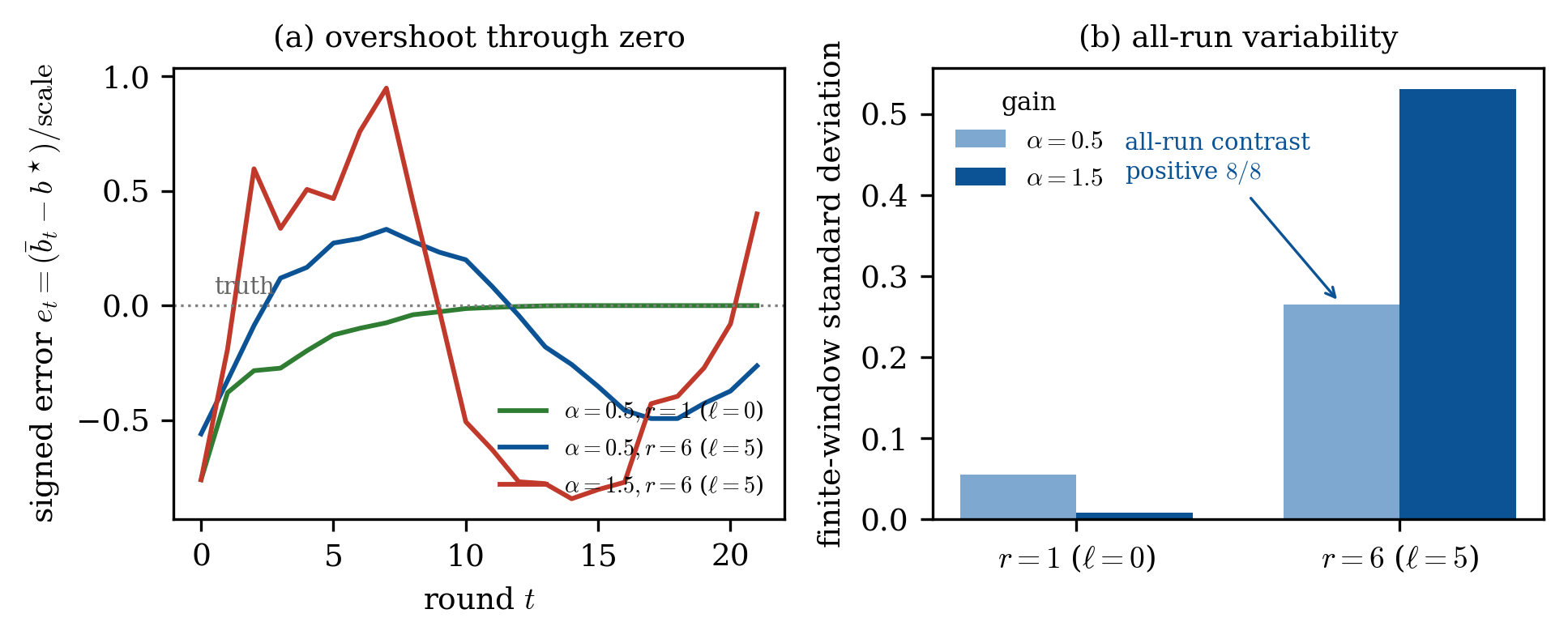}
  \caption{\textbf{Qwen3.6-35B under an externally computed controller.}
  (a)~Selected signed-error trajectories from the saved runs.
  (b)~All-run mean finite-window standard deviation for each gain and stored label $r$;
  the actual lag is $\ell=r-1$. At $\alpha=0.5$, all eight question-level means increase from
  $r=1$ to $r=6$. This is the exploratory all-run comparison; the movement-conditioned result is
  reported separately. Neither panel classifies asymptotic stability from finite observations.}
  \label{fig:realexp}
\end{figure}

\section{Discussion}

In the postulated recurrence, verification is a delayed control action with a finite gain margin.
Corollary~\ref{cor:dose-explicit} says a verifier can be \emph{too}
aggressive or \emph{too} slow; Corollary~\ref{cor:equal} gives the exact combined-gain ceiling when
communication and verification latencies coincide. These results give a model-specific reason to
check feedback dynamics when changing verifier strength or latency. Placement affects both stability
and tracking in the stated model. Our greedy guarantee concerns coherence reduction,
while the stability comparison for nested hard-corrector sets follows from the slowest grounded mode.
This complements online cascade \emph{detectors} \cite{CASPIAN,AgentForesight,GUARDIAN} with a
\emph{controller}-side stability guarantee that detection alone does not provide.

\begin{remark}[non-negative error and absorption]\label{rem:absorb}
The empirical contrast in Section~\ref{sec:emp} does not follow from the sign of the error measure
alone. Consider the illustrative projected recurrence
\[
  p_{t+1}=\max\!\big(0,\,p_t-\alpha\,\sigma(p_{t-\delta})+g\big),
\]
with $\alpha>0$, integer $\delta\ge0$, non-negative initial history, and continuous $\sigma$ satisfying
$\sigma(0)=0$ and $\sigma(p)>0$ for $p>0$. If $g=0$, the sequence for $t\ge0$ is non-increasing and converges to
zero. Indeed, a positive limit $\ell$ would leave an asymptotically positive decrement
$\alpha\sigma(\ell)$, contradicting convergence. Once zero is reached it is absorbing, including
while the delayed history is still positive.

If $g>0$, however, an all-zero history produces $p_{t+1}=g$, and non-negative periodic orbits are
possible. For example, with $\alpha=1$, $\delta=1$, $g=1$ and
$\sigma(p)=\max(-2,\min(p,2))$, the history $(p_{-1},p_0)=(0,0)$ generates
$0,0,1,2,2,1,0,0,\ldots$, a period-six orbit. With $\sigma=\tanh$, the parameters
$\alpha=1.5$, $\delta=1$, $g=0.2$ also produce recurrent non-negative oscillations in the
finite numerical check; exact periodicity is not proved for that example.
Thus projection removes signed overshoot but does not establish absence of oscillation under forcing.
The factual-QA observations should be interpreted as protocol-specific evidence; the projected
recurrence is not an identified model of those experiments. Likewise, loss of linear stability in
the signed model alone does not prove existence of a bounded nonlinear limit cycle.
\end{remark}

\section{Limitations}

Several limitations bound the scope. The analysis is exact for the postulated linear recurrence.
With constant forcing, its stable equilibrium is generally biased away from truth. Interpreting the
recurrence as a local model of a nonlinear LLM loop requires an independently identified equilibrium
and Jacobian; the synthetic onset comparison does not supply that identification. Unlike the dose limit, the placement rule is so far validated only on the coherence surrogate, not with LLM agents. Its approximation guarantee does not cover static mean-square error or the stochastic variance of the delayed loop.
We model constant additive fault forcing, hard pins and a separate delayed restoring term;
the placement surrogate uses finite soft pins. Byzantine faults and
edge-level (message-filtering) correctors are future work. For general unequal delays,
Theorem~\ref{thm:twodelay} gives candidate boundary loci, including singular cases; it does not provide
a closed-form classification of every stable cell or a universal worst delay allocation.
The eigen-decoupling assumes a \emph{symmetric} (reciprocal) interaction graph. Directed grounded
Laplacians can be non-normal, but directedness alone does not imply non-normality. For example,
$2I-P$, where $P$ is a three-cycle permutation matrix, is nonsymmetric and normal: it describes a
directed cycle with an additional unit truth-pinning edge at each free node. It has complex
eigenvalues, so the real-mode formulas do not directly apply even in this normal case.
Non-normal cases can additionally exhibit transient amplification. Directed extensions require
separate stability and transient analyses; the symmetric thresholds are not guaranteed conservative.

The controlled numeric-estimation study is limited by eight questions, three repetitions per cell,
a short observation window, the imposed correction law, the legacy indexing and metadata discrepancies,
and the mismatch between movement-conditioned and all-run analyses. Its finite-window standard
deviation includes transients; a stronger validation would identify dynamics independently and test
persistence at longer horizons. Neither grounded factual-QA
study quantitatively identifies the proposed mechanism. The historical AR diagnostic used the
stored delay labels as actual lags and classified a high constant final error as observed instability.
The previously quoted precision and recall were classification metrics, not coefficient-recovery
accuracy. We withdraw that interpretation. At effective lag zero, current and delayed regressors
are identical, so their separate coefficients are unidentifiable even with unlimited samples.
The lag-aware diagnostic now reports only the combined coefficient there; short, quantized,
often constant trajectories also limit inference at positive lag. The absorption argument applies
only to the unforced projected model (Remark~\ref{rem:absorb}).

The expanded factual study supplies complete logs and corrected indexing, but tests one server
model artifact, two delays, one faulty-peer count and a single trajectory per question and condition.
Its format and evidence differ from the historical protocol, and its primary grid has no verifier-off
control. Source review does not certify every reference, and post-hoc same-assistant semantic review
is not independent human validation. The $80$-character limit can leave unfinished statements even
without a token-limit termination. Repeated questions from related TruthfulQA families motivate
cluster resampling, but the chosen clustering does not guarantee independence across families.
The expanded sample size is not a demonstrated power calculation for the revised metric.
Stronger claims require a prospectively specified outcome that addresses abstentions, an explicit
effect size of interest, score validation, replication on another model and longer horizons; the
existing data should not be reclassified or selectively extended until a desired sign appears.

Finally, heterogeneous correction gains do not automatically preserve the modal reduction.
For a diagonal gain matrix $K$, independent grounded-Laplacian modes require a common eigenbasis
(equivalently, $[K,\Lg]=0$ for these symmetric matrices). This fails in general. Common stochastic
delays with scalar gains may preserve the spatial modes, but require a time-varying stability analysis
rather than a fixed characteristic polynomial. Coupling correction to an online change detector
likewise remains outside the present results.

\section{Conclusion}
We modeled the multi-agent LLM verifier loop as a delayed consensus over a graph with grounded
corrector nodes, and derived exact thresholds showing that verification carries a stability
budget in this model. Excessive correction gain can cause loss of linear stability.
When communication and verification delays coincide, the ceiling for the dimensionless combined
dose $\eta(\mu+\kappa)$ is the inverse golden ratio at delay two. This does not make equal delays universally worst. For a fixed corrector budget, greedy placement guarantees at least
$1-1/e$ of the optimal coherence reduction; this does not imply optimal static mean-square error.

The synthetic loop brackets onset near the linear prediction at the sweep resolution. The LLM
experiment demonstrates transmission of an externally specified correction signal. Correcting the
history indexing maps its stored labels $1,6$ to effective lags $0,5$; the exploratory all-run
comparison supports larger variability at the longer lag in five of six models. This does not
establish the same law for natural factual verification or persistent nonlinear oscillation.
The historical factual-QA contrasts are inconclusive and have small selected samples, aliased
delay labels and departures from the planned controls. The separate expanded study addresses
sample count, delay indexing and missing response logs, but its $400$ questions still leave both
signs of the binary-amplitude contrast possible under conservative completion bounds. Abstentions
dominate the remaining unscored cells. A post-hoc outcome analysis finds lower abstention frequency
at lag five on TruthfulQA, identically under both scoring versions; intervals for the change in
labelled-error variability span zero. This does not establish greater accuracy, absence of an
amplitude effect or general factual stability.

Within the assumed recurrence, dose is limited by stability,
while a placement objective must specify the error or noise model it represents.

\section*{Acknowledgments}
The author thanks Ilya Makarov for valuable feedback on the manuscript.

\appendix
\renewcommand{\thesection}{Appendix~\Alph{section}}
\section{Proof of Proposition~\ref{prop:boundary} (oscillatory boundary)}\label{app:proofs}
Fix $0<a<1$. At $\beta=0$ the roots of $P(z)=z^\delta(z-a)+\beta$ are $0$ and $a$, strictly
inside the unit disk. For $z=e^{i\theta}$, $0\le\theta\le\pi$, put
$\phi_a(\theta)=\arg(e^{i\theta}-a)$, taking the continuous argument from $0$ to $\pi$. Then
\[
  \phi_a'(\theta)=\frac{1-a\cos\theta}{1-2a\cos\theta+a^2}>0,
  \qquad \beta=|e^{i\theta}-a|=\sqrt{1+a^2-2a\cos\theta}.
\]
Positive $\beta$ requires $\delta\theta+\phi_a(\theta)=(2j+1)\pi$. The left side is strictly
increasing; the smallest admissible angle $\theta_\star$ is its unique solution at $j=0$.
Since $\partial\phi_a/\partial a=\sin\theta/(1-2a\cos\theta+a^2)>0$ for $0<\theta<\pi$,
the limiting arguments at $a=0$ and $a=1$ give $\theta<\phi_a(\theta)<(\pi+\theta)/2$. Hence
$\pi/(2\delta+1)<\theta_\star<\pi/(\delta+1)$. The modulus above increases strictly with
$\theta$, so this angle and its conjugate give the first crossing; any other admissible angle,
including $\theta=\pi$ when allowed, has larger dose. Separating real and imaginary parts of
$P(e^{i\theta_\star})=0$ gives \eqref{eq:boundary}.

Every unit-circle root is simple, since
$P'(z)=z^{\delta-1}((\delta+1)z-\delta a)$ cannot vanish there. Implicit differentiation at such
a root yields
\[
 \frac{\mathrm d\log|z|}{\mathrm d\beta}
 =\Re\frac{z-a}{\beta((\delta+1)z-\delta a)}
 =\frac{\delta+1+\delta a^2-(2\delta+1)a\cos\theta}
 {\beta\,| (\delta+1)z-\delta a |^2}>0.
\]
Indeed, the numerator is at least $(1-a)(\delta+1-\delta a)>0$. All crossings are outward,
so stability cannot return at larger $\beta$. This proves the stated necessary-and-sufficient
criterion for every integer $\delta\ge1$, without a finite-delay enumeration.

\section{Proof of Lemma~\ref{lem:cheb} (monotonicity of the dose)}\label{app:cheb}
Write $a_\delta(c)=U_\delta/U_{\delta-1}$. The Chebyshev recurrence $U_\delta=2cU_{\delta-1}-U_{\delta-2}$
gives a continued fraction and the derivative recursion
\[
  a_\delta=2c-\frac{1}{a_{\delta-1}},
  \qquad
  a_\delta'=2+\frac{a_{\delta-1}'}{a_{\delta-1}^2} .
\]
For $\delta\ge2$, the relevant branch has $c>\cos(\pi/\delta)$, to the right of the largest
zero of $U_{\delta-1}$ and of all preceding polynomials. Thus the denominators in this induction
are nonzero. Since $a_1'=2$, it follows that $a_\delta'>0$ on this branch.
All zeros of $U_{\delta-1}$ are smaller than $c$, and its leading coefficient is positive; the
product representation therefore gives $U_{\delta-1}>0$ and $U_{\delta-1}'>0$. Hence
\[
  \frac{\mathrm d\beta_c}{\mathrm da}=-\frac{U_{\delta-1}'}{U_{\delta-1}^2\,a_\delta'}<0 ,
\]
which proves strict decrease in the instantaneous coefficient $a$ for $\delta\ge2$.
At $\delta=1$, $U_0=1$ and $\beta_c=1$, so this derivative is zero.

For delay monotonicity, fix $a\in(0,1)$ and use the phase equation from~\ref{app:proofs}.
If $\theta_\delta$ solves $\delta\theta+\phi_a(\theta)=\pi$, then
$(\delta+1)\theta_\delta+\phi_a(\theta_\delta)=\pi+\theta_\delta>\pi$.
Strict increase of the phase implies $\theta_{\delta+1}<\theta_\delta$.
Since $\sqrt{1+a^2-2a\cos\theta}$ increases strictly with $\theta$, it follows that
$\beta_c(a,\delta+1)<\beta_c(a,\delta)$ for every integer $\delta\ge1$.

At $a=0$, direct factorization gives $\beta_c=1$ for all delays. As $a\uparrow1$, the phase equation
gives $\theta_\star\to\pi/(2\delta+1)$ and
$\beta_c(1,\delta)=2\sin(\pi/(4\delta+2))$. At $\delta=2$, eliminating $c$ gives
$\beta_c^2+a\beta_c=1$, hence the positive solution in the lemma.

\section{Two-delay locus and equal-delay ceiling}\label{app:twodelay}
For Theorem~\ref{thm:twodelay}, substitute $z=e^{i\theta}$ into $P(z)=0$ and divide by $z^h$.
The real and imaginary parts are \eqref{eq:twodelay}. Viewed as equations for $(p,q)$, their
coefficient matrix has determinant $\sin((\delta-d)\theta)$. If it is nonzero, Cramer's rule gives
the displayed parametrization. If it vanishes, the rank-one system must instead be solved for
consistency; it may contribute a line rather than a parametrized point. Direct substitution at
$z=1$ gives $p+q=0$, and at $z=-1$ gives
\[
  p\,(-1)^d+q\,(-1)^\delta=2 .
\]
No other unit-circle roots are possible. Root counts are locally constant away from this locus,
by continuity of roots and the fixed leading coefficient of $P$. Consequently each connected open
cell can be classified by counting roots at one representative point. At $d=0$, the regular curve
agrees with Proposition~\ref{prop:boundary} on its branch $0<1-p<1$. The equality $d=\delta$ is a
separate integer-delay case, not a continuous limit through the singular parametrization.

For Corollary~\ref{cor:equal}, write $b=p+q>0$, so the polynomial becomes
$z^{\tau+1}-z^\tau+b$. At $b=0$ the root at $1$ is simple and has derivative $\mathrm dz/\mathrm db=-1$;
all other roots are at zero. Thus sufficiently small positive $b$ is stable. For a non-real
unit-circle root with $0<\theta<\pi$, the phase equation is
$\tau\theta+(\pi+\theta)/2=(2j+1)\pi$ and $b=2\sin(\theta/2)$.
The first admissible angle is $\theta_\star=\pi/(2\tau+1)$, yielding the claimed threshold.
At any unit-circle crossing with $b>0$, the radial derivative from~\ref{app:proofs} with $a=1$
has numerator $(2\tau+1)(1-\cos\theta)>0$; the roots are simple and cross outward. Stability
therefore does not return at higher gain. Taking the largest $\mu$ gives the network condition.

\section{Proof of Theorem~\ref{thm:place} (greedy placement)}\label{app:place}
Throughout, write $M(R)=G+\kappa I+w\sum_{i\in R}u_iu_i^\top\succ0$ and $H(R)=\tr M(R)^{-1}$,
with the fixed baseline $G$ of Section~\ref{sec:place}.

\emph{Monotonicity.} For $i\notin R$, pinning is a positive-semidefinite update,
\[
  M(R\cup\{i\})=M(R)+w\,u_iu_i^\top\succeq M(R)
  \;\Longrightarrow\;
  M(R\cup\{i\})^{-1}\preceq M(R)^{-1},
\]
so $H(R\cup\{i\})\le H(R)$. Thus $H$ is non-increasing and $\rho(R)=H(\varnothing)-H(R)$ is non-decreasing
with $\rho(\varnothing)=0$.

\emph{Marginal gain.} By the Sherman--Morrison identity,
\[
  \bigl(M+w\,u_iu_i^\top\bigr)^{-1}=M^{-1}-\frac{w\,M^{-1}u_iu_i^\top M^{-1}}{1+w\,u_i^\top M^{-1}u_i},
\]
whose trace is $\tr M^{-1}-w\|M^{-1}u_i\|^2/(1+w\,u_i^\top M^{-1}u_i)$. Subtracting recovers
\eqref{eq:marg},
\[
  \Delta_i(R)=\frac{w\,\|M(R)^{-1}u_i\|^2}{1+w\,u_i^\top M(R)^{-1}u_i}\ \ge\ 0 .
\]

\emph{Submodularity.} We must show the marginal is non-increasing in the pinned set,
\[
  \Delta_i(R)\ge\Delta_i(S)\qquad\text{for } R\subseteq S,\ i\notin S .
\]
Here the structure of $G$ is essential. Because $M(R)=G+\kappa I+W_R$ has nonpositive off-diagonal
entries (from $-A$) and is positive definite, it is a symmetric nonsingular \emph{M-matrix}, so
\[
  M(R)^{-1}\ge0\quad\text{entrywise}.
\]
This inverse sign can be checked directly: choose $s>\lambda_{\max}(M)$ and write
$B=I-M/s$. Then $B$ is entrywise non-negative, its eigenvalues lie in $[0,1)$, and
$M^{-1}=s^{-1}\sum_{j=0}^{\infty}B^j$ is entrywise non-negative. The same argument applies
after any non-negative diagonal update.
For a direct diminishing-returns proof, let $Q(x)=(G+\kappa I+\diag x)^{-1}$ and
$f(x)=\tr Q(x)$ for $x\ge0$. Matrix differentiation gives
\[
 \partial_i Q=-Qu_iu_i^\top Q,\qquad
 \partial_i f=-(Q^2)_{ii},\qquad
 \partial_j\partial_i f=2Q_{ij}(Q^2)_{ij}\ge0.
\]
The last sign follows from entrywise non-negativity and symmetry of $Q$.
Thus $-\partial_i f$ is non-increasing as any other pin weight grows. Integrating it from $0$ to
$w$ along coordinate $i$ proves $\Delta_i(R)\ge\Delta_i(S)$ for $R\subseteq S$, $i\notin S$.
The Loewner order alone would not prove this sign for an arbitrary positive-definite matrix.
The numerical sweep in \texttt{placement\_demo.py --sweep} finds $0/4000$ violations for
Laplacian-plus-diagonal matrices and $115/4000$ for arbitrary positive-definite baselines.

\emph{Greedy bound.} Write $S_j$ for the greedy set after $j$ steps, $S_0=\varnothing$ in this
argument, and let $R^\star$ be an optimal size-$k$ soft-pin set. Monotonicity and submodularity give
\[
 \rho(R^\star)-\rho(S_j)
 \le\sum_{i\in R^\star\setminus S_j}\Delta_i(S_j)
 \le k\bigl(\rho(S_{j+1})-\rho(S_j)\bigr).
\]
Thus the remaining gap after a step is at most $1-1/k$ times the previous gap. Iterating from
$\rho(\varnothing)=0$ gives the finite-$k$ bound in \eqref{eq:greedybound};
$(1-1/k)^k\le e^{-1}$ gives its last inequality \cite{Nemhauser}. Substituting
$\rho(R)=H(\varnothing)-H(R)$ yields the stated bound on $H$.

\section{Reproducibility}\label{app:repro}
Repository: \url{https://github.com/YehudaItkin/delayed-verification-llm}.
For offline reproduction, use Python 3.11 or later and run from the repository root:
\begin{verbatim}
python -m pip install -r requirements-reproduce.txt
python reproduce.py
\end{verbatim}
This checks \path{RESEARCH_MANIFEST.json}, restores gzip-compressed logs and recomputes the
expanded analyses without querying a model. For restoration alone, run
\texttt{python restore\_data.py}. Paths below refer to restored files. Missing historical
raw generations still prevent independent reconstruction of some old scores and parser outputs.

Targeted Lean~4 certificates are under \path{paper/lean/}:
\path{PlacementCheck.lean} covers the placement counterexample and reduction-bound algebra;
\path{DelayCheck.lean} the delay-allocation and singular-locus examples;
\path{ControllerCheck.lean} archived indexing;
\path{RecheckCheck.lean} gossip-only, averaging and identifiability witnesses; and
\path{RereadCheck.lean} the normal directed matrix and fixed-peer expansion.
They do not certify the complete stability or placement theorems or the empirical model.

Expanded-study paths are relative to \path{experiments/factual_expanded_v4/}:
\path{code_v4_3/PROTOCOL_RU.md} specifies collection;
\path{qwen38/v4_3_snapshot/results/primary/manifest.json} freezes questions, evidence,
conditions, model digest and code hashes. The adjacent \path{requests.jsonl} and
\path{runs.jsonl} preserve all requests, responses and states;
\path{primary_v4_3_verification.json} records the technical replay.
Semantic decisions and their application are in \path{semantic_review_20260915/}
(v1) and \path{semantic_review_v2_20260915/} (v2).
\path{outcome_analysis_20260915/} contains the post-hoc protocol, code, tests, per-question
results, hashes and $20000$ bootstrap replicates per dataset and version, with seeds
$20260915$ (PsiloQA) and $20260916$ (TruthfulQA).
These records verify calculations conditional on frozen labels, not their independent factual validity.

\end{document}